\documentclass{llncs}
\usepackage{slashbox}

\usepackage{epsfig}
\usepackage{amstext}
\usepackage{amsfonts}
\usepackage{amsmath}
\usepackage{amssymb}
\usepackage{stmaryrd}
\usepackage{verbatim}
\usepackage{subfigure}

\usepackage{bibunits}

\usepackage[lined,ruled]{algorithm2e}

\usepackage[latin1]{inputenc}

\usepackage{xcolor,pgf,pgflibraryarrows,tikz}

\usetikzlibrary{fit} % fitting shapes to coordinates
\usetikzlibrary{backgrounds} % drawing the background after the foreground

\tikzstyle{background}=[rectangle,fill=gray!10, inner sep=0.1cm, rounded corners=0mm]
\usetikzlibrary{automata}
\usetikzlibrary{snakes,backgrounds,positioning,shadows}
\usetikzlibrary{patterns}
\usetikzlibrary{shapes}

\newcommand{\outcome}{\textnormal{\textsf{outcome}}}

\newcommand{\nat}{\mathbb{N}}

\newcommand{\strat}[1]{\langle\!\langle #1\rangle\!\rangle} 

\newcommand{\safe}{\textnormal{\textsf{Safe}}}

\newcommand{\plays}{\textnormal{\textsf{Plays}}}
\newcommand{\prefplays}{\textnormal{\textsf{PrefPlays}}}

\SetKw{Next}{next}

\newcommand{\always}{\Box}
\newcommand{\eventually}{\Diamond}

\newcommand{\Obs}{{\sf Obs}}

\newcommand{\obs}{{\sf obs}}

\title{Doomsday Equilibria for Omega-Regular Games}

\author{K. Chatterjee$^1$ \and L. Doyen$^2$ \and E. Filiot$^3$ \and J.-F. Raskin$^3$}

\institute{$^1$IST Austria, $^2$LSV-ENS de Cachan,\ $^3$CS-Universit\'e Libre de Bruxelles -- U.L.B.}

\begin{document}

\begin{bibunit}[abbrv]

\maketitle

\begin{abstract}
Two-player games on graphs provide the theoretical framework
for many important problems such as reactive synthesis.
While the traditional study of two-player zero-sum games has been 
extended to multi-player games with several notions of equilibria,
they are decidable only for perfect-information games,
whereas several applications require imperfect-information games.

In this paper we propose a new notion of equilibria, 
called doomsday equilibria, which is a strategy profile such that 
all players satisfy their own objective, and if any coalition of players 
deviates and violates even one of the players objective, then the objective
of every player is violated.

We present algorithms and complexity results for deciding 
the existence of doomsday equilibria for various classes of 
$\omega$-regular objectives, both for imperfect-information games, 
and for perfect-information games. We provide optimal complexity
bounds for imperfect-information games, and in most cases 
for perfect-information games.
\end{abstract}

\section{Introduction}

Two-player games on finite-state graphs with $\omega$-regular 
objectives provide the framework to study many important 
problems in computer science~\cite{Sha53,Rabin69,EJ91}. 
One key application area is synthesis of reactive systems~\cite{BL69,RW87,PR89}.
Traditionally, the reactive synthesis problem is reduced to 
two-player zero-sum games, where vertices of the graph represent 
states of the system, edges represent transitions, 
one player represents a component of the system to synthesize, and 
the other player represents the purely adversarial coalition of all the 
other components.
Since the coalition is adversarial, the game is zero-sum, i.e., 
the objectives of the two players are complementary.
Two-player zero-sum games have been studied in great depth in 
literature~\cite{Mar75,EJ91,automata}.

Instead of considering all the other components as purely adversarial, a more realistic 
model is to consider them as individual players each with their own objective,
as in protocol synthesis where the rational behavior of the agents 
is to first satisfy their own objective in the protocol before trying 
to be adversarial to the other agents.
Hence, inspired by recent applications in protocol synthesis, the model of 
multi-player games on graphs has become an active area of research in 
graph games and reactive synthesis~\cite{AHK02,FKL10,UW11}.
In a multi-player setting, the games are not necessarily zero-sum (i.e.,
objectives are not necessarily conflicting) and the classical notion of 
rational behavior is formalized as Nash equilibria~\cite{Nas50}. 
Nash equilibria perfectly capture the notion of rational behavior 
in the absence of external criteria, i.e., the players are concerned
only about their own payoff (internal criteria), and they are indifferent to the payoff of the
other players. 
In the setting of synthesis, the more appropriate notion is the adversarial
external criteria, where the players are as harmful as possible to the other
players without sabotaging with their own objectives. 
This has inspired the study of refinements of Nash equilibria, such as 
secure equilibria~\cite{CHJ06} (that captures the adversarial external criteria), 
rational synthesis~\cite{FKL10}, and led to several new logics  
where the non-zero-sum equilibria can be expressed~\cite{CHP10,DLM10,MMV10,WHY11,MMPV12}.
The complexity of Nash equilibria~\cite{UW11}, secure equilibria~\cite{CHJ06}, rational 
synthesis~\cite{FKL10}, and of the new logics has been studied recently~\cite{CHP10,DLM10,MMV10,WHY11}.
%%and the complexities range from polynomial time to undecidability.

Along with the theoretical study of refinements of equilibria, applications 
have also been developed in the synthesis of protocols. 
In particular, the notion of secure equilibria has been useful in the synthesis 
of mutual-exclusion protocol~\cite{CHJ06}, and of fair-exchange protocols~\cite{KR03,CKS06} 
(a key protocol in the area of security for exchange of digital signatures).
One major drawback that all the notions of equilibria suffer is that 
the basic decision questions related to them are decidable only in the setting 
of perfect-information games (in a perfect-information games the players perfectly 
know the state and history of the game, whereas in imperfect-information games each player
has only a partial view of the state space of the game), and 
in the setting of multi-player imperfect-information games they are undecidable~\cite{PR89}.
However, the model of imperfect-information games is very natural 
because every component of a system has private variables not accessible to other components,
and recent works have demonstrated that imperfect-information games are required in 
synthesis of fair-exchange protocols~\cite{JMM11}. In this paper, we provide the first
decidable framework that can model them.

We propose a new notion of equilibria which we call 
\emph{doomsday-threatening} equilibria (for short, doomsday equilibria).
A doomsday equilibria is a strategy profile such that all players satisfy
their own objective, and if any coalition of players deviates and violates 
even one of the players objective, then doomsday follows (every player objective
is violated). Note that in contrast to other notions of equilibria, doomsday equilibria
consider deviation by an arbitrary set of players, rather than individual
players.
Moreover, in case of two-player non-zero-sum games they coincide with the 
secure equilibria~\cite{CHJ06} where objectives of both players are satisfied.

\begin{figure}
\label{FEPs}
  \begin{center}
	\includegraphics{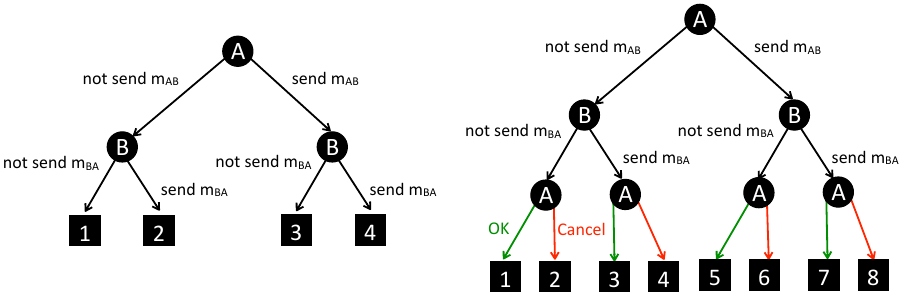}
  \end{center}
\caption{A simple example in the domain of Fair Exchange Protocols}
\end{figure}

\begin{example}
Let us consider the two trees of Fig.~\ref{FEPs}. They model the possible behaviors of two entities {\sf Alice} and {\sf Bob} that have the objective of exchanging messages: ${\sf m_{AB}}$ from {\sf Alice} to {\sf Bob}, and ${\sf m_{BA}}$ from {\sf Bob} to {\sf Alice}. Assume for the sake of illustration that ${\sf m_{AB}}$ models the transfer of property of an house from {\sf Alice} to  {\sf Bob}, while ${\sf m_{BA}}$ models the payment of the price of the house from {\sf Bob} to {\sf Alice}. 

Having that interpretation in mind, let us consider the left tree. On the one hand, {\sf Alice} has as primary objective (internal criterion) to reach either state $2$ or state $4$, states in which she has obtained the money, and she has a slight preference for $2$ as in that case she received the money while not transferring the property of her house to {\sf Bob}, this corresponds to her adversarial external criterion. On the other hand, {\sf Bob} would like to reach either state $3$ or $4$ (with again a slight preference for $3$). 
Also, it should be clear that {\sf Alice} would hate to reach $3$ because she would have transferred the property of her house to {\sf Bob} but without being paid. Similarly, {\sf Bob} would hate to reach $2$. To summarize, {\sf Alice} has the following preference order on the final states of the protocol: $2 > 4 > 1 > 3$, while for {\sf Bob} the order is $3 > 4 > 1 > 2$. Is there a {\em doomsday threatening equilibrium} in this game ? For such an equilibrium to exist, we must find a pair of strategies that please the two players for their primary objective (internal criterion): reach $\{2,4\}$ for {\sf Alice} and reach $\{1,2\}$ for {\sf Bob}. Clearly, this is only possible if at the root {\sf Alice} plays "{\sf send} ${\sf m_{AB}}$", as otherwise we would not reach $\{1,2\}$ violating the primary objective of {\sf Bob}. But playing that action is not safe for {\sf Alice} as {\sf Bob} would then choose  "{\sf not send} ${\sf m_{AB}}$" because he slightly prefers $3$ to $4$. It can be shown that the only rational way of playing (taking into account both internal and external criteria) is for {\sf Alice} to play "{\sf not send} ${\sf m_{AB}}$" and for {\sf Bob} would to play "{\sf not send} ${\sf m_{BA}}$". This way of playing is in fact the only secure equilibrium of the game but this is not what we hope from such a protocol. 

The difficulty in this exchange of messages comes from the fact that {\sf Alice} is starting the protocol by sending her part and this exposes her. To obtain a better behaving protocol, one solution is to add an additional stage after the exchanges of the two messages as depicted in the right tree of Fig.~\ref{FEPs}. In this new protocol, {\sf Alice} has the possibility to cancel  the exchange of messages (in practice this would be implemented by the intervention of a {\sf TTP}\footnote{{\sf TTP} stands for \emph{Trusted Third Party}}). For that new game, the preference orderings of the players are as follows: for {\sf Alice} it is $3>7>1=2=4=6=8>5$, and for {\sf Bod} it is $5>7>1=2=4=6=8>3$. Let us now show that there is a doomsday equilibrium in this new game. In the first round, {\sf Alice} should play "{\sf send} ${\sf m_{AB}}$" as otherwise the internal objective of {\sf Bob} would be violated, then {\sf Bob} should play "{\sf send} ${\sf m_{BA}}$", and finally {\sf Alice} should play ``{\sf OK}'' to validate the exchange of messages. Clearly, this profile of strategies satisfies the first property of a doomsday equilibrium: both players have reached their primary objective. Second, let us show that no player has an incentive to deviate from that profile of strategies.  First, if {\sf Alice} deviates then {\sf Bob} would play "{\sf not send} ${\sf m_{BA}}$", and we obtain a doomsday situation as both players have their primary objectives violated. Second, if {\sf Bob} deviates by playing "{\sf not send} ${\sf m_{BA}}$", then {\sf Alice} would cancel the protocol exchange which again produces a doomsday situation. So, no player has an incentive to deviate from the equilibrium and the outcome of the protocol is the desired one: the two messages have been fairly exchanged. So, we see that the threat of a doomsday that brings the action "{\sf Cancel}" has a beneficial influence on the behavior of the two players.$\hfill \square$
\end{example}

It should now be clear that multi-player games with 
doomsday equilibria provide a suitable framework to model various problems in protocol synthesis.
In addition to the definition of doomsday equilibria, our main contributions are to present algorithms and complexity bounds for deciding the existence of 
such equilibria for various classes of $\omega$-regular objectives both 
in the perfect-information and in the imperfect-information cases. In all cases but one, we establish the exact complexity.
Our technical contributions are summarized in Table \ref{table:summary}. More specifically:
\begin{enumerate}

\item \emph{(Perfect-information games).}
We show that deciding the existence of doomsday equilibria in 
multi-player perfect-information games is (i)~PTIME-complete for 
reachability, B\"uchi, and coB\"uchi objectives; (ii)~PSPACE-complete for safety objectives;
and (iii)~in PSPACE and both NP-hard and coNP-hard for parity objectives.

\item \emph{(Imperfect-information games).}
We show that deciding the existence of doomsday equilibria in 
multi-player imperfect-information games is EXPTIME-complete for 
reachability, safety, B\"uchi, coB\"uchi, and parity objectives.

\end{enumerate}
The area of multi-player games and various notion of equilibria is 
an active area of research, but notions that lead to decidability 
in the imperfect-information setting and has applications in synthesis
has largely been an unexplored area. 
Our work is a step towards it.

\begin{table}[t]
\centering
\begin{tabular}{|c||c|c|c|c|c|}
\hline
objectives & safety & reachability & B¸chi & co-B¸chi & parity \\
\hline 
 & & & & & \textsc{PSpace} \\
perfect information & \textsc{PSpace-C} & \textsc{PTime-C} & \textsc{PTime-C} & \textsc{PTime-C} & \textsc{NP-Hard} \\
 & & & & & \textsc{CoNP-Hard} \\
\hline 
imperfect information & \textsc{ExpTime-C} & \textsc{ExpTime-C} & \textsc{ExpTime-C} & \textsc{ExpTime-C}  & \textsc{ExpTime-C} \\
\hline
\end{tabular}
\vspace{2mm}
\caption{\label{table:summary} Summary of the results}
\vspace{-10mm}
\end{table}

\section{Doomsday Equilibria for Perfect Information Games}\label{sec:prelim}

In this section, we define game arena with perfect information, $\omega$-regular objectives,
and doomsday equilibria. 

\paragraph{{\bf Game Arena}} 
An $n$-player game arena $G$ with perfect information
is defined as a tuple $(S, {\cal P},s_{{\sf init}}, \Sigma, \Delta)$ such that 
$S$ is a nonempty finite set of \textit{states}, ${\cal P}=\{S_1,S_2,\dots,S_n\}$ is a partition of $S$ into $n$ classes of states, one for each player respectively, $s_{{\sf init}} \in S$ is the initial state, $\Sigma$ is
a finite set of actions, and $\Delta : S \times \Sigma \rightarrow S$ is the transition function.

Plays in $n$-player game arena $G$ are constructed as follows. They start in the initial state $s_{{\sf init}}$, and then an $\omega$ number of rounds are played as follows: the player that owns the current state $s$ chooses a letter $\sigma \in \Sigma$ and the game evolves to the position $s' = \Delta(s,\sigma)$, then a new round starts from $s'$. So formally, a {\em play} in $G$ is an infinite sequence $s_0 s_1 \dots s_n \dots$ such that $(i)$ $s_0=s_{{\sf init}}$ and $(i)$ for all $i \geq 0$, there exists $\sigma \in \Sigma$ such that $s_{i+1}=\Delta(s_i,\sigma)$. 
The set of plays in $G$ is denoted by $\plays(G)$, and the set of finite prefixes of plays by $\prefplays(G)$. We denote by $\rho,\rho_1,\rho_i, \dots$ plays in $G$, by $\rho(0..j)$ the prefix of the play $\rho$ up to position $j$ and by $\rho(j)$ the position $j$ in the play $\rho$. We also use $\pi$, $\pi_1$, $\pi_2$, ... to denote prefixes of plays. Let $i \in \{1,2,\dots,n\}$, a prefix $\pi$ belongs to Player~$i$ if ${\sf last}(\pi)$, the last state of $\pi$, belongs to Player~$i$, i.e. ${\sf last}(\pi) \in S_i$. We denote by $\prefplays_i(G)$ the set of prefixes of plays in $G$ that belongs to Player~$i$.

\paragraph{{\bf Strategies and strategy profiles}} 
A {\em strategy} for Player~$i$, for $i \in \{1,2,\dots,n\}$, is a mapping 
$\lambda_i :  \prefplays_i(G) \rightarrow \Sigma$ from prefixes of plays to actions. %
A {\em strategy profile} $\Lambda = (\lambda_1,\lambda_2,\dots,\lambda_n)$ is a tuple of strategies such that $\lambda_i$ is a strategy of Player~$i$. The strategy of Player~$i$ in 
$\Lambda$ is denoted by $\Lambda_i$, 
and the the tuple of the remaining strategies $(\lambda_1,\dots,\lambda_{i-1},\lambda_{i+1},\dots,\lambda_n)$ by  $\Lambda_{-i}$. 
For a strategy $\lambda_i$ of Player $i$, we define its {\em outcome} as the set of plays that are consistent with $\lambda_i$: formally, 
$\outcome_i(\lambda_i)$ is the set of $\rho \in \plays(G)$ such that for all $j \geq 0$, if $\rho(0..j) \in \prefplays_i(G)$, then $\rho(j+1)=\Delta(\rho(j),\lambda_i(\rho(0..j)))$.
Similarly, we define the {\em outcome of a strategy profile} $\Lambda=(\lambda_1,\lambda_2,\dots,\lambda_n)$, as the unique play $\rho \in \plays(G)$ such that for all positions $j$, for all $i \in \{1,2,\dots,n\}$, if $\rho(j) \in \prefplays_i(G)$ then $\rho(j+1)=\Delta(\rho(j),\lambda_i(\rho(0..j)))$.
Finally, given 
a state $s\in S$ of the game, we denote by $G_s$ the game $G$ whose initial state is replaced by $s$.

\paragraph{{\bf Winning objectives}}
A {\em winning objective} (or an \emph{objective} for short) $\varphi_i$ for Player $i {\in} \{1,2,\dots,n\}$ is a set of infinite sequences of states, i.e. $\varphi_i {\subseteq} S^{\omega}$. 
A strategy $\lambda_i$ is {\em winning} for Player $i$ (against all other players) w.r.t. an objective
$\varphi_i$ if $\outcome_i(\lambda_i)\subseteq \varphi_i$. 

%\marginpar{{\bf JF}: Do we need to introduce the notion of Nash equilibrium here ?}
%For a tuple of winning conditions $\overline{\varphi} =
%(\varphi_1,\dots,\varphi_n)$, and a strategy profile $\overline{\sigma}$,
%the value of $\overline{\sigma}$, denoted by
%$v_{\overline{\varphi}}(\overline{\sigma})$, is an $n$-tuple in
%$\{0,1\}^n$ defined by $(v_{\overline{\varphi}}(\overline{\sigma}))_i =
%1$ iff $\overline{\sigma}_i$ is winning for Player $i$ wrt $\varphi_i$.
%
%
%\paragraph{Nash Equilibria} A strategy profile $\overline{\sigma} = (\sigma_1,\dots,\sigma_n)$ is
%a \textit{Nash Equilibrium} if for all $i\in\{1,\dots,n\}$, for all
%strategy $\lambda_i$ of Player $i$:
%$$
%(v_{\overline{\varphi}}(\sigma_1,\dots,\sigma_{i-1},\sigma_{i},\sigma_{i+1},\dots,\sigma_n))_i \geq (v_{\overline{\varphi}}(\sigma_1,\dots,\sigma_{i-1},\lambda_i,\sigma_{i+1},\dots,\sigma_n))_i
%$$

%
Given an infinite sequence of states $\rho \in S^{\omega}$, we denote by ${\sf visit}(\rho)$ the set of states that appear at least once along $\rho$, i.e. ${\sf visit}(\rho)=\{ s \in S | \exists i \geq 0\cdot \rho(i)=s \}$, and ${\sf inf}(\rho)$ the set of states that appear infinitely often along $\rho$, i.e. ${\sf inf}(\rho)=\{ s \in S | \forall i \geq 0\cdot \exists j \geq i \cdot \rho(i)=s \}$.
We consider the following types of winning objectives:
  \begin{itemize}
  	\item a {\em safety objective} is defined by a subset of states $T \subseteq S$ that has to be never left: ${\sf safe}(T)=\{ \rho \in S^{\omega} \mid {\sf visit}(\rho) \subseteq T \}$;  
  	\item a {\em reachability objective} is defined by a subset of states $T \subseteq S$ that has to be reached: ${\sf reach}(T)=\{ \rho \in S^{\omega} \mid {\sf visit}(\rho) \cap T \not= \emptyset \}$; 
	\item a {\em B\"uchi objective} is defined by a subset of states $T \subseteq S$ that has to be visited infinitely often: ${\sf B\ddot uchi}(T)=\{ \rho \in S^{\omega} \mid {\sf inf}(\rho) \cap T \not= \emptyset \}$; 
	\item a {\em co-B\"uchi objective} is defined by a subset of states $T \subseteq S$ that has to be reached eventually and never be left: ${\sf coB\ddot uchi}(T)=\{ \rho \in S^{\omega} \mid {\sf inf}(\rho) \subseteq T \}$; 
	\item let $d \in \nat$, a {\em parity objective with $d$ priorities} is 
          defined by a priority function $p : S \rightarrow  \{0,1,\dots,d\}$ as the set of plays such that the smallest priority visited infinitely often is even: ${\sf parity}(p)=\{ \rho \in S^{\omega} | \min \{ p(s) \mid s \in {\sf inf}(\rho)\} \mbox{~is~even} \}$.
  \end{itemize}
B\"uchi, co-B\"uchi and parity objectives $\varphi$ are called {\em tail objectives} because they enjoy the following closure property: for all $\rho \in \varphi$ and all $\pi \in S^{*}$, $\rho\in \varphi$ iff $\pi \cdot \rho \in \varphi$.

Finally, given an objective $\varphi\subseteq S^\omega$ and a subset $P\subseteq \{1,\dots,n\}$, we
write $\strat{P} \varphi$ to denote the set of states $s$ from which the players from $P$ can cooperate
to enforce $\varphi$ when they start playing in $s$. Formally, 
$\strat{P} \varphi$ is the set of states $s$ such that there exists a set of strategies 
$\{ \lambda_i\ |\ i\in P\}$ in $G_s$, one for each player in $P$, such that
 $\bigcap_{i\in P} \outcome_i(\lambda_i) \subseteq \varphi$.

\paragraph{{\bf Doomsday Equilibria}} A strategy profile $\Lambda=(\lambda_1,\lambda_2,\dots,\lambda_n)$ is a 
{\em doomsday-threatening equilibrium} (doomsday equilibrium or DE for short) if:
\begin{enumerate}
  \item  it is winning for all the players, i.e. $\outcome(\Lambda) \in\bigcap_i \varphi_i$;

  \item each player is able to retaliate in case of deviation: for all $1 \leq i \leq n$, for all $\rho\in \outcome_i(\lambda_i)$, if
    $\rho\not\in \varphi_i$, then $\rho \in \bigcap_{j=1}^{j=n} \overline{\varphi_j}$ (doomsday), where $\overline{\varphi_j}$ denotes the complement of $\varphi_j$ in $S^\omega$.  
\end{enumerate}

In other words, when all players stick to their strategies then they all win, and if any arbitrary coalition of players deviates and makes
even just one other player lose then this player retaliates and ensures a doomsday, i.e. all players lose. 

\paragraph{Relation with Secure Equilibria}
In two-player games, the doomsday equilibria coincide with the notion of secure equilibrium~\cite{CHJ06} where both players satisfy
their objectives. In secure equilibria, for all $i\in\{1,2\}$, any
deviation of  Player $i$ that does not decrease her payoff does not decrease the payoff of Player $3{-}i$ either. In other words, if a deviation of Player $i$ decreases (strictly)
the payoff of Player $3{-}i$, i.e. $\varphi_{3{-}i}$ is not satisfied, then it also decreases her own payoff, i.e. $\varphi_i$ is not satisfied. A two-player secure equilibrium where both
players satisfy their objectives is therefore a doomsday equilibrium.

%% \paragraph{Remark} We could also require that at least one of the
%% Player can punish any selfish Player:

%% \begin{enumerate}
%%   \item[4.] \textbf{there exists} $i$, $\outcome(\lambda_i)\models (\bigvee_{i\neq
%%     j} W_j)\rightarrow W_i$. 
%% \end{enumerate}

\begin{figure*}[!t] 

\begin{center}

\begin{tabular}{cc}
\begin{minipage}[b]{0.45\linewidth}

\subfigure[Doomsday (Safety)]{
  \label{fig:example1}

\begin{tikzpicture}[very thick, scale=0.77]
\tikzstyle{every state}=[fill=gray!20!white]

\tikzset{P3/.style={diamond,draw,very thick,minimum size=7mm}}
\tikzset{P2/.style={rectangle,draw,very thick,minimum size=7mm}}
\tikzset{P1/.style={circle,draw,very thick,minimum size=7mm}}
\tikzset{P1init/.style={initial,circle,draw,very thick,minimum size=7mm}}
\tikzset{P1final/.style={accepting,circle,draw,very thick,minimum size=7mm}}

\node[P1](g1) at (0,0) {\rotatebox{270}{:-)}} ;
\node[P2](g2) at (2,0) {\rotatebox{270}{:-)}} ;
\node[P3](g3) at (4,0) {\rotatebox{270}{:-)}} ;
\node[P2](b1) at (0,2) {\rotatebox{270}{:-(}} ;
\node[P3](b2) at (2,2) {\rotatebox{270}{:-(}} ;
\node[P1](b3) at (4,2) {\rotatebox{270}{:-(}} ;

\draw [->] (g1) to (g2);
\draw [->] (g2) to (g3);
\draw [->] (g3) to [bend left=45] (g1);

\draw [->] (b1) to (b2);
\draw [->] (b2) to (b3);
\draw [->] (b3) to [bend right=45] (b1);

\draw [->] (g1) to (b1);
\draw [->] (g2) to (b2);
\draw [->] (g3) to (b3);

% \path (sinit') edge [loop above] node {$(\#,\#,\#)$} (sinit');

% \draw[line width=.5mm,color=black] (11,4.7) node[sttstates] {$(1,0,0)$};

\end{tikzpicture}
}
\end{minipage}

\begin{minipage}[b]{0.45\linewidth}

\subfigure[Büchi objectives]{
 \label{fig:example2}

\begin{tikzpicture}[very thick, scale=0.8]
\tikzstyle{every state}=[fill=gray!20!white]

\tikzset{P3/.style={diamond,draw,very thick,minimum size=10mm}}
\tikzset{P2/.style={rectangle,draw,very thick,minimum size=7mm}}
\tikzset{P1/.style={circle,draw,very thick,minimum size=7mm}}
\tikzset{P1init/.style={initial,circle,draw,very thick,minimum size=7mm}}

\node[P1](g1) at (0,0) {$s_1$} ;
\node[P2](g2) at (2,0) {$s_2$} ;
\node[P3](g3) at (1,-1.5) {$s_3$} ;
\node[P2](b1) at (0,-3) {\rotatebox{270}{:-)}} ;
\node[P1](b2) at (2,-3) {\rotatebox{270}{:-)}} ;

\node[P2](gg3) at (-1.8,0.2) {\rotatebox{270}{:-)}} ;
\node[P3](g4) at (-0.3,1.8) {\rotatebox{270}{:-)}} ;
\node[P1](g5) at (3.8,0.2) {\rotatebox{270}{:-)}} ;
\node[P3](g6) at (2.3,1.8) {\rotatebox{270}{:-)}} ;

\draw [->] (g1) to (g2);
\draw [->] (g2) to (g3);
\draw [->] (g3) to (g1);

\draw [->] (b1) to (g3);
\draw [->] (b2) to (b1);
\draw [->] (g3) to (b2);

\draw [->] (g2) to (g6);
\draw [->] (g6) to (g5);
\draw [->] (g5) to (g2);

\draw [->] (g1) to (gg3);
\draw [->] (gg3) to (g4);
\draw [->] (g4) to (g1);

% \path (sinit') edge [loop above] node {$(\#,\#,\#)$} (sinit');

% \draw[line width=.5mm,color=black] (11,4.7) node[sttstates] {$(1,0,0)$};

\end{tikzpicture}
}
\end{minipage}
\end{tabular}
\end{center}
\caption{\label{fig:ex1} Examples of doomsday equilibria for Safety and B\"uchi objectives}
\end{figure*}
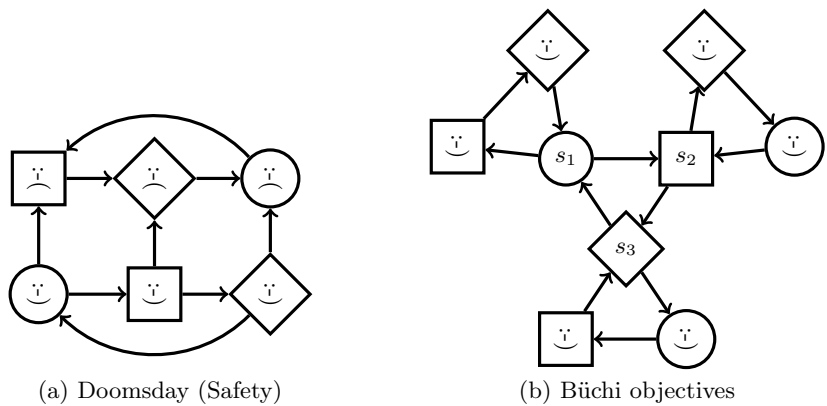

\begin{example} Fig. \ref{fig:ex1} gives two examples of games with safety and B\"uchi objectives respectively. Actions are in bijection with edges so they are not represented.

(Safety) Consider the 3-player game arena with perfect information of Fig. \ref{fig:example1} and
safety objectives. Unsafe states for each player are given by the respective nodes of
the upper part. Assume that the initial state is one of the safe states. This example
models a situation where three countries are in peace until one of the countries,
say country $i$, decides to attack country $j$. This attack will then necessarily be 
followed by a doomsday situation: country $j$ has a strategy to punish all other countries.
The doomsday equilibrium in this example is to play safe for all players.

(B\"uchi)  Consider the 3-player game arena with perfect information of Fig. \ref{fig:example2} with Büchi objectives for
each player: Player $i$ wants to visit infinitely often one of its ``happy'' states. The position of
the initial state does not matter. To make things more concrete, let us use this game to model
a protocol where 3 players want to share in each round a piece of information made of three parts: for all 
$i\in\{1,2,3\}$,  Player $i$ knows information $i\ mod\ 3 + 1$ and $i\ mod\ 3 +2$. 
Player $i$ can send or not
these informations to the other players. This is modeled by the fact that Player $i$ can decide
to visit the happy states of the other players, or move directly to
$s_{(i\ mod\ 3)+1}$. The objective of each player is to have an infinite
number of successful rounds where they get all information.

There are several doomsday equilibria. 
As a first one, let us consider the situation where for all $i\in\{1,2,3\}$, if Player $i$ is in state $s_i$, first it visits the happy states, and
when the play comes back in $s_i$, it moves to $s_{(i\ mod\ 3) + 1}$. This defines an infinite play that 
visits all the states infinitely often. Whenever some player deviates from this play, the other players retaliate by always choosing in the future to go to the next $s$ state instead of taking their respective loops. 
Clearly, if all players follow their respective strategies all happy states
are visited infinitely often. Now consider the strategy of Player $i$ against two strategies of the other players that makes
him lose. Clearly, the only way Player $i$ loses is when the two other players eventually never take their states, but then
all the players lose.

As a second one, consider the strategies where Player $2$ and Player $3$ always take their loops but Player $1$ never takes
his loop, and such that whenever the play deviates, Player $2$ and $3$ retialate by never taking their loops. For the same reasons as
before this strategy profile is a doomsday equilibrium.

Note that the first equilibrium requires one bit of memory for
each player, to remember if they visit their $s$ state for the first or second times. In the second equilibrium, only 
Player $2$ and $3$ needs a bit of memory. An exhaustive analysis shows that there is no memoryless doosmday equilibrium in this example.$\hfill \square$
\end{example}

%Consider the following specification of a mutual
%exclusion protocol:
%$$
%\bigwedge_{i=1}^n (\always \eventually a_i)\wedge \always
%(a_i\rightarrow \bigwedge_{j\neq i} \neg a_j)
%$$
%
%A doomsday equilibrium for this spec is defined by the strategies
%$\lambda_i$ which count up to $n$ and hold the token
%at rank $i$ and releases it at rank $i+1$. If some player has changed his strategy, it may create
%a conflict. To punish all the other players you just have to hold the
%token all the time after you got the first conflict.

%%% Local Variables: 
%%% mode: latex
%%% TeX-master: t
%%% End: 

\section{Complexity of DE for Perfect Information Games}\label{sec:perfect}

In this section, we prove the following results:

\begin{theorem}
    The problem of deciding the existence of a doomsday equilibrium in an $n$-player perfect information game arena and $n$ objectives $(\varphi_i)_{1\leq i\leq n}$ is:
    \begin{itemize}
        \item {\sc PTime-C} if the objectives $(\varphi_i)_{1\leq i\leq n}$ are either all B\"uchi, all co-B\"uchi or all reachability objectives,
        \item {\sc NP-hard}, {\sc coNP-hard} and in {\sc PSPace} if $(\varphi_i)_{1\leq i\leq n}$ are parity objectives,
        \item {\sc PSPace-C} if $(\varphi_i)_{1\leq i\leq n}$ are safety objectives.
    \end{itemize}
\end{theorem}

\noindent In the sequel, game arena with perfect information
are just called game arena.

\paragraph{{\bf Tail objectives}}
We first present a generic algorithm that works for any tail objective and then analyze its complexity for the different cases. Then we establish the lower bounds.
Let us consider the following algorithm:
  \begin{itemize}
  	\item compute the retaliation region of each player: $R_i = \strat{i} (\varphi_i \cup \bigcap_{j=1}^{j=n} \overline{\varphi_j})$;
	\item check for the existence of a play within $\bigcap_{i=1}^{i=n} R_i$ that satisfies all the objectives $\varphi_i$.
  \end{itemize}

\noindent The correctness of this generic procedure is formalized in the following lemma:
  
\begin{lemma}\label{lem:retaliateregion}
Let $G=(S,{\cal P},s_{\sf init},\Sigma,\Delta)$ be an $n$-player game arena with $n$ tail objectives $(\varphi_i)_{1\leq i\leq n}$. Let $R_i = \strat{i} (\varphi_i \cup \bigcap_{j=1}^{j=n} \overline{\varphi_j})$
be the retaliation region for Player~$i$. There is a doomsday equilibrium in $G$ iff there exists an infinite play that 
$(1)$ belongs to $\bigcap_{i=1}^{i=n} \varphi_i$ and
$(2)$ stays within the set of states $\bigcap_{i=1}^{i=n} R_i$ and  
\end{lemma}
\begin{proof}
First, assume that there exists an infinite play $\rho$ such that $\rho \in \bigcap_i (\varphi_i \cap R_i^\omega)$. From $\rho$, and the retaliating strategies that exist in all states of $R_i$ for each player, we show the existence of DE $\Lambda=(\lambda_1,\lambda_2,\dots,\lambda_n)$. Player~$i$ plays strategy $\lambda_i$ as follows:
 he plays according to the choices made in $\rho$ as long as all the other players do so, and
as soon as the play deviates from $\rho$, Player~$i$ plays his retaliating strategy (when it is his turn to play).

First, let us show that if Player~$j$, for some $j \not=i$, deviates and the turn comes back to Player~$i$ in a state $s$ then $s \in R_i$. Assume that Player~$j$ deviates when he is in some $s'\in S_j$. As before there was no deviation, by definition of $\rho$, $s'$ belongs to $R_i$. But no matter what the adversary are doing in a state that belongs to $R_i$, the next state must be a state that belongs to $R_i$ (there is only the possibility to leave $R_i$ when Player~$i$ plays). So, by induction on the length of the segment of play that separates $s'$ and $s$, we can conclude that $s$ belongs to $R_i$. From $s$, Player~$i$ plays a retaliating
strategy and so all the outcomes from $s$ are in $\varphi_i\cup \bigcap_{j=1}^{j=n}\overline{\varphi_j}$, and since the objective are tails, the prefix up to $s$ is not important and we get
(from $s_{{\sf init}}$) $\outcome_i(\lambda_i)\subseteq \varphi_i\cup \bigcap_{j=1}^{j=n}\overline{\varphi_j}$. Therefore the second property of the definition of doomsday equilibria is satisfied. Hence $\Lambda$ is a DE.

Let us now consider the other direction. Assume that $\Lambda$ is a DE. Then let us show that $\rho=\outcome(\Lambda)$
satisfies properties $(1)$ and $(2)$. By definition of DE, we know that $\rho$ is winning for all the players, so $(1)$ is satisfied.
Again by definition of DE, $\outcome(\Lambda_i)\subseteq \varphi_i\cup \bigcap_{j=1}^{j=n} \overline{\varphi_j}$. Let $s$ be
a state of $\rho$ and $\pi$ the prefix of $\rho$ up to $s$. For all outcomes $\rho'$ of $\Lambda_i$ in $G_s$, 
we have $\pi\rho'\in \varphi_i\cup \bigcap_{j=1}^{j=n} \overline{\varphi_j}$, and since the objectives are tail, we get
$\rho'\in \varphi_i\cup \bigcap_{j=1}^{j=n} \overline{\varphi_j}$. Hence $s\in R_i$. Since this property holds
for all $i$, we get $s\in \bigcap_i R_i$, and $(2)$ is satisfied. \qed
\end{proof}

\noindent Accordingly, we obtain the following upper-bounds:

\begin{lemma}
The problem of deciding the existence of a doomsday equilibrium in an $n$-player game arena can be decided in {\sc PTime} for B\"uchi and co-B\"uchi objectives, and in {\sc PSpace}
 for parity objectives.
\end{lemma}
\begin{proof}
By Lemma \ref{lem:retaliateregion} one first needs to compute the retaliation regions $R_i$ for all
$i\in\{1,\dots,n\}$. Once the sets $R_i$ have been computed, it is clear that the existence of a play winning  for all players is decidable in {\sc PTime} for all the three types of objectives.
For the B\"uchi and the co-B\"uchi cases, let us show how to compute the retaliation regions $R_i$. We start with B\"uchi and we assume that each player wants to visit
a set of states $T_i$ infinitely often. Computing the sets $R_i$ boils down to computing the set of states
$s$ from which Player $i$ has a strategy to enforce the objective (in LTL syntax) $\always\Diamond T_i\ \vee\ \bigwedge_{j=1}^{j=n} \Diamond \always \overline{T_j}$, which is equivalent to the formula
$\always\Diamond T_i\ \vee\ \Diamond \always \bigcap_{j=1}^{j=n} \overline{T_j}$. This is equivalent to
a disjunction of a B\"uchi and a co-B\"uchi objective, which is thus equivalent to a Streett objective
with one Streett pair and can be solved in PTime with a classical algorithm, e.g. \cite{DBLP:conf/lics/PitermanP06}. Similarly, for co-Büchi objectives, one can reduce the computation of the regions $R_i$ in polynomial time to the disjunction of a Büchi objective and a co-Büchi objective.

For the parity case, the winning objectives for the retaliation sets can be encoded compactly as Muller objectives
defined by a propositional formula using one proposition per state. Then they can be solved in 
{\sc PSpace} using the algorithm of Emerson and Lei presented in~\cite{EmersonL85}.\qed
\end{proof}

%We now consider {\em parity} games. Let $G=(S, s_{{\sf init}}, (\Sigma)_{1 \leq i \leq n}, \Delta)$ be an $n$-player game arena and $d \in \nat$.
%%
%In parity games, the objectives are given by {\em parity functions} $f_i : S \rightarrow \{0,1,\dots,d\}$, one for each player $i \in \{1,2,\dots,n\}$.
%%
%Given an infinite play $\rho=s_0 s_1 \dots s_n \dots$ in $G$, we denote by ${\sf Inf}(\rho)$ for the set of $s \in S$ that appear infinitely many times along $\rho$. Given a parity function $f$, a play $\rho$ satisfies the parity objective defined by $f$, written $\rho \models f$, if $\min \{ f(s) \mid s \in {\sf Inf}(\rho) \}$ is {\em even}.
%% 
%For parity objectives, doomsday equilibria can be reformulated as follows:
%a strategy profile $\Lambda=(\lambda_1,\lambda_2,\dots,\lambda_n)$ witnesses a doomsday equilibrium the game arena $G$ and parity objectives $(f_i)_{i=1,n2,\dots,n}$ if:
%\begin{enumerate}
%  \item for all $i$, $\outcome(\Lambda)\models f_i$; 
%  \item for all $i$, $\outcome(\lambda_i) \not\models f_i$ implies for all $j \not=i$, $\outcome(\lambda_i) \not\models f_j$
%\end{enumerate}
\noindent Let us now establish the lower bounds.

\begin{lemma}\label{lem:lowerboundstail}
The problem of deciding the existence of a DE in an $n$-player game arena is {\sc PTime-Hard} for B\"uchi and co-B\"uchi objectives, {\sc NP-Hard} and {\sc coNP-Hard} for parity objectives. All the hardness results hold even for a fixed  number of players.
\end{lemma}

\begin{proof}
   The hardness for B\"uchi and co-B\"uchi objectives holds already for 2 players. We describe the reduction for B\"uchi and it is similar for co-B\"uchi. We reduce the problem of deciding the winner in a two-player zero-sum game arena $G$ with a B\"uchi objective (known as a {\sc PTime-Hard} problem \cite{Imm4}) to the existence of a DE for B\"uchi objectives with two players. Consider a copy $G'$ of the game arena $G$
   and the following two objectives: Player 1 has the same B\"uchi objective as Player $1$ in $G$, and Player $2$ has a trivial B\"uchi objective (i.e. all states are B\"uchi states). Then clearly there exists
   a DE in $G'$ iff Player 1 has a winning strategy in $G$. Details are given in appendix.

For parity games, we can reduce zero-sum two-player games with a conjunction of parity objectives (known to be {\sc coNP-Hard} \cite{ChatterjeeHP07}) to the existence of a DE in a three player game with parity objectives.  
Similarly, we can reduce the problem of deciding the winner in a two-player zero-sum game with a disjunction of parity objectives (known to be {\sc NP-Hard} \cite{ChatterjeeHP07}) to the existence of a DE in a two-player game with parity objectives. The main idea in the two cases is to construct a game arena where one of the players can retaliate iff Player 1 in the original two-player zero-sum game has a winning strategy.
Details are given in appendix. \qed
\end{proof}
 
 % \begin{corollary}
 As a corollary of this result, deciding the existence of a secure equilibrium in a
$2$-player game such that both players satisfy their parity objectives is  {\sc NP-Hard}.
 % \end{corollary}

%We prove in the following lemma that the existence of a doomsday equilibrium for parity objectives can be establish with an algorithm that uses polynomial space only.
%
%\begin{lemma}\label{lem:vneparity}
%The existence of a doomsday equilibrium in an $n$-player game with perfect information and parity winning conditions can be established with an algorithm that uses polynomial space.
%\end{lemma}
%\begin{proof}
%The algorithm for deciding the existence of a doomsday equilibrium for parity objective is similar to the one for safety objectives. The only difference is when computing the set $R_i$ from which Player~i can ensure that he has a strategy 
%$\outcome(\lambda_i) \not\models f_i$ implies for all $j \not=i$, $\outcome(\lambda_i) \not\models f_j$.  This condition can be encoded compactly as a Muller condition defined by a proposition formula using one proposition per state and solved in polynomial space using the algorithm of Emerson and Lei defined in~\cite{EmersonL85}.\qed
%\end{proof}

\paragraph{{\bf Reachability objectives}}

We now establish the complexity of deciding the existence of a doomsday equilibria
in an $n$-player game with reachability objectives. We first establish an important property for
reachability objectives: 

\begin{proposition}
\label{lem:allalone}
Let $G=(S, {\cal P},s_{{\sf init}}, \Sigma, \Delta)$ be a game arena, and $(T_i)_{1\leq i\leq n}$ be $n$ subsets of $S$.
Let $\Lambda$ be a doomsday equilibrium in $G$ for the reachability objectives $({\sf Reach}(T_i))_{1\leq i\leq n}$. Let $s$ the first state in $\outcome(\Lambda)$ such that $s\in \bigcup_i T_i$. Then every player
has a strategy from $s$, against all the other players, to reach his target set.
\end{proposition}
\begin{proof} 
W.l.o.g. we can assume that $s\in T_1$. If some player, say Player $2$, as no strategy from $s$ to reach his target set $T_2$, then
necessarily $s\not\in T_2$ and by determinancy the other players have a strategy from $s$ to make Player $2$ lose. This contradicts the
fact that $\Lambda$ is a doomsday equilibrium as it means that $\Lambda_2$ is not a retaliating strategy. \qed
% W.l.o.g. assume that $s \in T_1$, and assume, again w.l.o.g. and for the sake of obtaining a contradiction that Player~$2$ does not have a strategy against all other players to force ${\sf reach}(T_2)$ from $s$. Then, we know that there exist a strategy profile $\Lambda'_{-2}=(\lambda_1,\lambda_3,\dots,\lambda_n)$ for the other players such that when they play according to this strategy profile, then $T_2$  is never visited. Then, if after the prefix that reaches $s$, those players play according to $\Lambda'_{-2}$ then the outcome $\rho$ of the game will be such that $\rho \not\in {\sf reach}(T_2)$ and $\rho \in \bigcap_{j \not= 2} {\sf reach}(T_j)$ which violates the condition~2 of the definition of doomsday equilibrium, and so contradicts the fact that $\Lambda=(\lambda_1,\lambda_2,\dots,\lambda_n)$ witnesses such an equilibrium.\qed
\end{proof}

\begin{lemma}
The problem of deciding the existence of a doomsday equilibrium in an $n$-player game with reachability objectives is in {\sc PTime}.
\end{lemma}
\begin{proof}The algorithm consists in:

$(1)$ computing the sets $R_i$ from which player $i$ can retaliate, i.e. the set of states $s$ from which Player~$i$ has a strategy to force, against all other players, an outcome such that $\eventually T_i \lor ( \bigwedge_{j=1}^{j=n} \always \overline{T_j})$. This set can be obtained
          by first computing the set of states $\strat{i} \eventually T_i$ from which Player~$i$ can force to reach $T_i$. It is done in {\sc PTime} by solving a classical two-player reachability game.
Then the set of states where Player~$i$ has a strategy $\lambda_i$ such that $\outcome_i(\lambda_i) \models \always ((\bigcap_{j=1}^{j=n} \overline{T_j}) \lor \strat{i} \eventually T_i) \}$, that is to confine the plays in states that do not satisfy the reachability objectives of the adversaries or from where Player~$i$ can force its own reachability objective. 
Again this can be done in {\sc PTime} by solving a classical two-player safety game.

$(2)$ then, checking the existence of some $i\in\{1,\dots,n\}$ and some finite path $\pi$ starting from $s_{{\sf init}}$ and that stays within $\bigcap_{j=1}^{j=n} R_j$ before reaching a state $s$ such that $s \in T_i$ and $s \in \bigcap_{j=1}^{j=n}\strat{j} \eventually T_j$.

Let us now prove the correctness of our algorithm. From its output, we can construct the strategy
profile $\Lambda$ where each $\Lambda_j$ ($j=1,\dots,n$) is as follows: follow $\pi$ up to the point
where either another player deviates and then play the retaliating strategy available in $R_i$, or to the point
where $s$ is visited for the first time and then play according to a strategy (from $s$) that force a visit to $T_i$ no matter
how the other players are playing. Clearly, $\Lambda$ witnesses a DE. Indeed, if $s$ is
reached, then all players have a strategy to reach their target set (including Player $i$ since $s\in T_i$) . By playing so they will all
eventually reach it. Before reaching $s$, if some of them deviate, the other have a strategy to retaliate as $\pi$ stays in $\bigcap_{j=1}^{j=n} R_j$. The other direction follows from Proposition~\ref{lem:allalone}. \qed
\end{proof}

\begin{lemma}
The problem of deciding the existence of a DE in a $2$-player game with reachability objectives is {\sc PTime-Hard}.
\end{lemma}
\begin{proof}
It is proved by an easy reduction from the And-Or graph reachability problem~\cite{Imm4}: if reachability
is trivial for one of the two players, the existence of a doomsday equilibrium is equivalent to
the existence of a winning strategy for the other player in a two-player zero sum reachability game.\qed
\end{proof}

\paragraph{{\bf Safety Objectives}}

%A {\em safety objective} is defined by sets of safe states $(\safe_i)_{i=1,n2,\dots,n} \subseteq S$, one for each player $i \in \{1,2,\dots,n\}$. Each set of states $\safe_i$ implicitly defines the set $W_i=\{ \rho \in \plays(G) \mid \forall j \geq 0 \cdot \rho(j) \in \safe_i \}$ of winning plays for Player~$i$. When a play $\rho$ satisfies a safety objective defined by the set of safe states $\safe$, we write $\rho \models \always \safe$.
%The notion of doomsday equilibrium for safety game can thus be rephrased as follows: a strategy
%profile $\Lambda=(\lambda_1,\lambda_2,\dots,\lambda_n)$ witnesses a doomsday equilibrium for the game arena $G$ and safety objectives $(\safe_i)_{i=1,n2,\dots,n}$ if:
%\begin{enumerate}
%  \item $\outcome(\Lambda)\models \bigwedge_i \always \safe_i$;
%  \item for all $i \in \{1,2,\dots,n\}$, $\outcome_i(\lambda_i)\models (\neg \always
%    \safe_i)\rightarrow (\bigwedge_{j\neq i} \neg \always \safe_j)$
%\end{enumerate}

% The next theorem states the exact complexity of deciding the existence of a doomsday equilibrium in a multi-player game with safety objectives.

We establish the complexity of deciding the existence of a doomsday equilibrium in an $n$-player game with perfect information and safety objectives.

\begin{lemma}[{\sc PSpace-Easyness}]
The existence of a doomsday equilibrium in an $n$-player game with safety objectives can be decided in {\sc PSpace}.
\end{lemma}

\begin{proof}
Let us consider an $n$-player game arena $G=(S,{\cal P},s_{\sf init},\Sigma,\Delta)$ and $n$ safety objectives
${\sf safe}(T_1),\dots,{\sf safe}(T_n)$ for $T_1\subseteq S,\dots,T_n\subseteq S$. 
The algorithm is composed of the following two steps:

$(1)$ For each Player~$i$, compute the set of states $s \in S$ in the game such that Player~$i$ can retaliate whenever necessary, i.e. the set of states $s$ from where there exists a strategy $\lambda_i$ for Player~$i$ such that $\outcome_i(\lambda_i)$ satisfies $\neg (\always T_i) \rightarrow  \bigwedge_{j=1}^{j=n} \neg \always T_j$, 
                or equivalently 			$\neg (\eventually \overline{T_i}) \vee  \bigwedge_{j=1}^{j=n} \eventually  \overline{T_j}$. 		This can be done in {\sc PSpace} using a result by Alur et al. (Theorem 5.4 of ~\cite{AluLaT04}) on solving two-player games whose
                Player 1's objective is defined by Boolean combinations of LTL formulas that use only $\eventually$ and $\wedge$.  We denote by $R_i$
                the set of states in $G$ where Player~$i$ has a strategy to retaliate. 

 % Indeed, Theorem~4 in~\cite{AlurTM03} tells us
		% % that we can compute in {\sc PSpace} the set of states $s$ from which Player~$i$ has a strategy $\lambda_i$ such that
		% \begin{equation}
			% \label{eq:simtoorig}
 			% % \outcome(\lambda_i)\models \bigvee_{j \neq i} \always \safe_j \rightarrow (\always \eventually \safe_i).
		% \end{equation}
		% \noindent
		% % It is easy to see that the formulas of eq.~\ref{eq:orig} and eq.~\ref{eq:simtoorig} are equivalent on game arena where the violation of $\safe_i$ is made 			
		% % permanent, i.e. whenever the set $\safe_i$ is left by a prefix of play then it is never reentered. 
		% % It is easy to modify the original game arena $G$ so that this property is true: add one bit of information to the state space to record 
		% the violation and declare  unsafe all the states where this bit is set to $1$. 	
		
		%For all $s \in R_i \cap S_i$, we say that
                %$\sigma \in \Sigma_i$ is \emph{good} for Player $i$ if $\Delta(s,\sigma) \in R_i$. 

	% \item We compute for each Player~$i$ the set of states $R'_i$ from which Player~$i$ has a choice of action 
		% $\sigma \in \Sigma$ such that for any actions played by the other players, the next state is in $R_i$. 
		% This set can be computed in polynomial time (and so in polynomial space). 
		% For $s \in R'_i$, we denote by ${\sf good}_i(s) \subseteq \Sigma_i$ the set of actions of Player~$i$ with this property.

	$(2)$ then, verify whether there exists an infinite path in $\bigcap_{i=1}^{i=n}({\sf safe}(T_i) \cap R_i)$.

\noindent
Now, let us establish the correctness of this algorithm.
Assume that an infinite path exists in $\bigcap_{i=1}^{i=n}({\sf safe}(T_i) \cap R_i)$. The strategies $\lambda_i$ for each Player~$i$ are defined as follows:  play the moves that are prescribed as long as every other players do so, and as soon as the play deviates from the infinite path, play the retaliating strategy.

It is easy to see that the profile of strategies $\Lambda=(\lambda_1,\lambda_2,\dots,\lambda_n)$ is a DE. Indeed, 
the states are all safe for all players as long as they play their strategies. Moreover, as before deviation the play is within $\bigcap_{i=1}^{i=n} R_i$, if Player~$j$ deviates, we know that the state that is reached after deviation is still in $\bigcap_{j=1}^{j=n} R_j$ and therefore the other players can retaliate.

Second, assume that $\Lambda=(\lambda_1,\lambda_2,\dots,\lambda_n)$ is a DE in the $n$-player game $G$ for the safety objectives $({\sf safe}(T_i))_{1\leq i\leq n}$. 
Let $\rho = \outcome(\lambda_1,\lambda_2,\dots,\lambda_n)$. By definition of doomsday equilibrium, we know that all states appearing in $\rho$ satisfy all the safety objectives, i.e. $\rho \models \bigwedge_{i=1}^{i=n} \always T_i$. Let us show that the play also remains within $\bigcap_{i=1}^{i=n} R_i$. Let $s$ be a state of $\rho$, $i\in\{1,\dots,n\}$, and 
$\pi$ the finite prefix of $\rho$ up to $s$. By definition of DE we have $\outcome(\lambda_i)\models \always T_i\vee \bigwedge_{j=1}^{j=n}\eventually \overline{T_j}$. 
Therefore for all outcomes $\rho'$ of $\lambda_i$ in $G_s$, 
$\pi\rho'\models \always T_i\vee \bigwedge_{j=1}^{j=n}\eventually \overline{T_j}$. Moreover,
$\pi\models \bigwedge_{j=1}^{j=n} \always T_j$ since it is a prefix of $\rho$. Therefore
$\rho'\models \always T_i\vee \bigwedge_{j=1}^{j=n}\eventually \overline{T_j}$ and $s\in R_i$. Since
it holds for all $i\in\{1,\dots,n\}$, we get $s\in \bigcap_{i=1}^{i=n} R_i$. \qed

\end{proof}

\begin{lemma}[{\sc PSpace-Hardness}]\label{lem:hardness-safety-perfect}
The problem of deciding the existence of a doomsday equilibrium in an $n$-player game with safety objectives is {\sc PSpace-Hard}.
\end{lemma}

\begin{proof}
We present a reduction from the problem of deciding the winner in a zero-sum two-player game with a conjunction of $k$ reachability
objectives (aka generalized reachability games), which is  a {\sc PSpace-C} problem~\cite{AlurTM03}. The idea of
the reduction is to construct a non-zero sum $(k+1)$-player game where one of the players has a retaliating strategy iff there is a winning
strategy in the generalized reachability game.\qed
\end{proof}

\section{Complexity of DE for Imperfect Information Games}

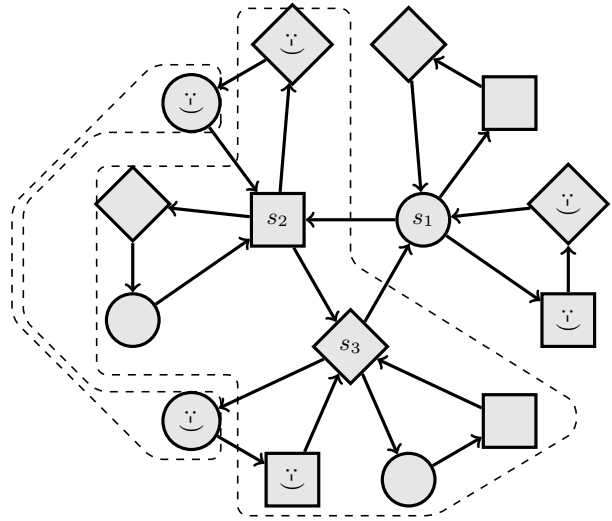
\begin{figure}[t]
\begin{center}
\begin{tikzpicture}[very thick, scale=0.75]
\tikzstyle{every state}=[fill=gray!20!white]

\tikzset{P3/.style={diamond,draw,very thick,minimum size=10mm,fill=gray!20!white}}
\tikzset{P2/.style={rectangle,draw,very thick,minimum size=7mm,fill=gray!20!white}}
\tikzset{P1/.style={circle,draw,very thick,minimum size=7mm,fill=gray!20!white}}
\tikzset{P1init/.style={initial,circle,draw,very thick,minimum size=7mm,fill=gray!20!white}}

% \node[P1](g1) at (0,0) {$s_1$} ;
% \node[P2](g2) at (2,0) {$s_2$} ;

% \node[P3](g3) at (1,-1.5) {$s_3$} ;

% \node[P2](b1) at (-1.3,-4.2) {:-)} ;
% \node[P1](b2) at (3.3,-4.2) {:-)} ;
% \node[P2](b3) at (0.3,-3.5) {} ;
% \node[P1](b4) at (1.7,-3.5) {} ;

% \node[P2](gg3) at (-3.3,-0.4) {:-)} ;
% \node[P3](g4) at (-1,3.3) {:-)} ;
% \node[P2](gg5) at (-2,0.3) {} ;
% \node[P3](gg6) at (-1.2,1.5) {} ;

% \node[P1](g5) at (5.3,-0.4) {:-)} ;
% \node[P3](g6) at (3,3.3) {:-)} ;
% \node[P1](gb5) at (4,0.3) {} ;
% \node[P3](gb6) at (3.2,1.5) {} ;

\node[P2](g12) at (3.86, -1.03) {\rotatebox{270}{:-)}} ;
\node[P3](g13) at (3.86, 1.03) {\rotatebox{270}{:-)}} ;
\node[P2](b12) at (2.82, 2.82) {} ;
\node[P3](b13) at (1.03, 3.86) {} ;

\node[P1](c1) at (1.29, 0.75) {$s_1$} ;

\node[P3](g23) at (-1.03, 3.86) {\rotatebox{270}{:-)}} ;
\node[P1](g21) at (-2.82, 2.82) {\rotatebox{270}{:-)}} ;
\node[P3](b23) at (-3.86, 1.03) {} ;
\node[P1](b21) at (-3.86, -1.03) {} ;

\node[P2](c2) at (-1.29, 0.75) {$s_2$} ;

\node[P1](g31) at (-2.82, -2.82) {\rotatebox{270}{:-)}} ;
\node[P2](g32) at (-1.03, -3.86) {\rotatebox{270}{:-)}} ;
\node[P1](b31) at (1.03, -3.86) {} ;
\node[P2](b32) at (2.82, -2.82) {} ;

\node[P3](c3) at (0, -1.5) {$s_3$} ;

\draw[semithick, dashed, rounded corners] (0,0) -- (0,4.5) -- (-2, 4.5) -- (-2,1.7) -- (-4.5, 1.7) -- (-4.5, -2) -- (-2, -2) -- (-2, -4.5) -- (1.5,-4.5) -- (4, -3) -- (4, -2.5) -- cycle;

\draw[semithick, dashed, rounded corners] (-2.3,2.3) -- (-2.3, 3.5) -- (-3.5, 3.5) -- (-6, 1) -- (-6, -1) -- (-3.5, -3.5) -- (-2.3, -3.5) -- (-2.3, -2.3) -- (-4.5, -2.3) -- (-5.8, -1) -- (-5.8,1) -- (-4.5, 2.3) -- (-2.3, 2.3) -- cycle;

% (0,0) --
  % ++(1, 0) -- ++(0,-1) -- ++(1, 0)  -- ++(1,0) -- ++(0,-1) --
  % ++(0,-1) -- ++(-1,0) -- ++(0,-1) -- ++(-1,0) -- ++(0, 1) --
  % ++(-1,0) -- ++(0, 1) -- ++(0, 1) --  cycle;

\draw [->] (c1) to (c2);
\draw [->] (c2) to (c3);
\draw [->] (c3) to (c1);

\draw [->] (c1) to (g12);
\draw [->] (g12) to (g13);
\draw [->] (g13) to (c1);

\draw [->] (c1) to (b12);
\draw [->] (b12) to (b13);
\draw [->] (b13) to (c1);

\draw [->] (c2) to (g23);
\draw [->] (g23) to (g21);
\draw [->] (g21) to (c2);

\draw [->] (c2) to (b23);
\draw [->] (b23) to (b21);
\draw [->] (b21) to (c2);

\draw [->] (c3) to (g31);
\draw [->] (g31) to (g32);
\draw [->] (g32) to (c3);

\draw [->] (c3) to (b31);
\draw [->] (b31) to (b32);
\draw [->] (b32) to (c3);

% \draw [->] (b3) to (g3);
% \draw [->] (b4) to (b3);
% \draw [->] (g3) to (b4);

% \draw [->] (g2) to (g6);
% \draw [->] (g6) to (g5);
% \draw [->] (g5) to (g2);

% \draw [->] (g2) to (gb6);
% \draw [->] (gb6) to (gb5);
% \draw [->] (gb5) to (g2);

% \draw [->] (g1) to (gg3);
% \draw [->] (gg3) to (g4);
% \draw [->] (g4) to (g1);

% \draw [->] (g1) to (gg5);
% \draw [->] (gg5) to (gg6);
% \draw [->] (gg6) to (g1);

% \path (sinit') edge [loop above] node {$(\#,\#,\#)$} (sinit');

% \draw[line width=.5mm,color=black] (11,4.7) node[sttstates] {$(1,0,0)$};

\end{tikzpicture}
\end{center}
\caption{\label{fig:example-imperfect} Game arena with imperfect information and Büchi objectives. Only undistinguishable states of Player $1$ (circle) are depicted. Observations are symmetric for the other players.}
\end{figure}

In this section, we define $n$-player game arenas with imperfect information. We adapt to this context the notions of observation, observation of a play, observation-based strategies, and we study the notion of doomsday equilibria when players are restricted to play observation-based strategies.

\paragraph{{\bf Game arena with imperfect information}} An $n$-player game arena with {\em imperfect
information} is a tuple $G = (S,{\cal P}, s_{{\sf init}}, \Sigma, \Delta, (O_i)_{1\leq i\leq n})$ such that
$(S,{\cal P}, s_{{\sf init}}, \Sigma, \Delta)$ is a game arena (of perfect information) and 
for all $i$, $1 \leq i \leq n$, $O_i\subseteq 2^S$ is a {\em partition} of $S$.  Each block
in $O_i$ is called an {\em observation} of Player~$i$. %We assume that for all $i\in\{1,\dots,n\}$, 
%for all $o\in O_i$, either $o\subseteq S_i$ or $o\cap S_i=\emptyset$, i.e. each player knows when it
%is his turn to play. 
We assume that %
  	% \begin{itemize}
		% % \item $S$ is the state space of the game, $(S_1,S_2,\dots,S_n)$ is a partition of $S$ where $S_i$, $1 \leq i \leq n$, is the set of states controlled by Player~$i$.
		% \item $s_{{\sf init}} \in S_1$ is the initial state of the game (it is controlled by Player 1), 
		% \item $\Sigma$ is the set of actions of the game, 
		% % % % % \item $\Delta : S \times \Sigma \rightarrow S$ is the transition function of the game and it respects the following property: for all $s \in S_i$, for all $\sigma \in \Sigma$, if $i<n$ then $\Delta(s,\sigma) \in S_{i+1}$, and if $i=n$ then $\Delta(s,\sigma) \in S_{1}$. So, the players play according to a fix order\footnote{This restriction is not necessary to obtain the results presented in this section but it makes some of our notations lighter.}.
		% % \item for all $i$, $1 \leq i \leq n$, $O_i=\{o^i_1,o^i_2,\dots,o^i_k\}$ is a {\em partition} of $S$, each block in this partition is called an {\em observation} of Player~$i$.
	% \end{itemize}
% \noindent
% Note that as in our game, 
the players play in a predefined order\footnote{This restriction is not necessary to obtain the results presented in this section (e.g. Theorem \ref{thm:imperfect}) but it makes some of our notations lighter.}: for all $i\in\{1,\dots,n\}$, all $q\in S_i$ and all $\sigma\in \Sigma$, $\Delta(q,\sigma)\in S_{(i\ mod\ n) + 1}$.

% % Given a prefix of a play $\pi$, the length of this prefix determines the player that will play the next action. This player is denoted by ${\sf turn}(\pi)$ and by
% ${\sf turn}(\pi)=(| \pi | \mod n) +1$.

\paragraph{{\bf Observations}}
For all $i\in\{1, \dots,n\}$, we denote by $O_i(s)\subseteq S$ the block in $O_i$ that
contains $s$, that is the observation that Player~$i$ has when he is in state $s$. 
We say that two states $s,s'$ are {\em undistinguishable} for
Player~$i$ if $O_i(s) = O_i(s')$. This defines an equivalence relation on states 
that we denote by $\sim_i$. 
The notions of plays and prefixes of plays are slight variations from the perfect information setting: a play in $G$ is a sequence $\rho=s_0,\sigma_0,s_1,\sigma_1,\dots \in (S \cdot \Sigma)^{\omega}$ such that $s_0=s_{{\sf init}}$, and for all $j \geq 0$, we have $s_{j+1}=\Delta(s_j,\sigma_j)$. A prefix of play is a sequence $\pi = s_0,\sigma_0,s_1,\sigma_1,\dots,s_k \in (S \cdot \Sigma)^{*} \cdot S$ that can be extended into a play. As in the perfect information setting, we use the notations $\plays(G)$ and $\prefplays(G)$ to denote the set of plays in $G$ and its set of prefixes, and $\prefplays_i(G)$ for the set of prefixes that end in a state that belongs to Player~$i$. While actions are introduced explicitly in our notion of play and prefix of play, their visibility is limited by the notion of observation. The {\em observation} of a play $\rho=s_0,\sigma_0,s_1,\sigma_1,\dots$ by Player~$i$ is the infinite sequence written $\Obs_i(\rho) \in ({O_i} \times (\Sigma \cup \{ \tau \})^{\omega}$ such that for all $j \geq 0$, $\Obs_i(\rho)(j)=(O_i(s_j),\tau)$ if $s_j \not\in S_i$, and $\Obs_i(\rho)(j)=(O_i(s_j),\sigma_j)$ if $s_j \in S_i$. Thus, only actions played by Player~$i$ are visible along the play, and the actions played by the other players are replaced by $\tau$. The observation $\Obs_i(\pi)$ of a prefix
$\pi$ is defined similarly. %
% % % % Accordingly, the observation of a prefix is as follows: let $\pi \in s_0,\sigma_0,s_1,\sigma_1,\dots,s_k \in \prefplays(G)$,  then $\Obs_i(\pi) \in ({O_i} \times (\Sigma \cup \{ \tau \})^{*} \cdot O_i$ such that for all $j$, $0 \leq j < k$, $\Obs_i(\rho)(j)=(O_i(s_j),\tau)$ if $s_j \not\in S_i$, and $\Obs_i(\rho)(j)=(O_i(s_j),\sigma_j)$ if $s_j \in S_i$, and $o_k=\Obs_i(s_k)$.
Given an infinite sequence of observations $\eta \in (O_i \times ( \Sigma \cup \{ \tau \} ))^{\omega}$ for Player~$i$, we denote by $\gamma_i(\eta)$ the set of plays in $G$ that are compatible with $\eta$, i.e. $\gamma_i(\eta)= \{ \rho \in \plays(G) \mid \Obs_i(\rho)=\eta \}$. The functions $\gamma_i$ are extended to prefixes of sequences of observations naturally.

\paragraph{{\bf Observation-based strategies and doomsday equilibria}} 
A strategy $\lambda_i$ of Player $i$ is {\em observation-based} if for all prefixes of plays $\pi_1,\pi_2 \in \prefplays_i(G)$ such that $\Obs_i(\pi_1)=\Obs_i(\pi_2)$,  it holds that $\lambda_i(\pi_1)=\lambda_i(\pi_2)$, i.e. while playing with an observation-based strategy, Player~$i$ plays the same action after undistinguishable prefixes.
A strategy profile $\Lambda$ is observation-based if each $\Lambda_i$ is observation-based.
Winning objectives, strategy outcomes and winning strategies are defined as in the
perfect information setting. We also define the notion of outcome
relative to a prefix of a play. Given an observation-based strategy $\lambda_i$ for 
Player $i$, and a prefix $\pi = s_0,\sigma_0,\dots, s_k\in \prefplays_i(G)$, the strategy 
$\lambda_i^\pi$ is defined for all prefixes $\pi'\in \prefplays_i(G_{s_k})$ where $G_{s_k}$ is the
game arena $G$ with initial state $s_k$, by $\lambda_i^{\pi}(\pi') = \lambda_i(\pi\cdot \pi')$. 
The set of outcomes of the strategy $\lambda_i$ \emph{relative to} $\pi$ is defined by $\outcome_i(\pi,\lambda_i) = \pi\cdot \outcome_i(\lambda^\pi_i)$.

% \{ \rho\in\outcome_i(\lambda_i)\ |\ \rho = \pi, \sigma_k,s_{k+1},\sigma_{k+1},\dots \}$. 

% \marginpar{JF: Those winning conditions should have been defined already}
% We consider three types of objectives for games of imperfect information:
  % \begin{itemize}
  	% % % \item	{\em Reachability objectives}: the objectives are given by $(R_1,R_2,\dots,R_n)$ where each $R_i \subseteq S$ implicitly defines the set of winning plays $$W_i=\{ \rho \in \plays(G) \mid \exists j \geq 0 \cdot \rho(j)=(s,\sigma) \land s \in R_i \}.$$
	% % % \item {\em Safety objectives}: the objectives are given by $({\sf Safe}_1,{\sf Safe}_2,\dots,{\sf Safe}_n)$ where each ${\sf Safe}_i \subseteq S$ implicitly defines the set of winning plays $$W_i=\{ \rho \in \plays(G) \mid \forall j \geq 0 \cdot \rho(j)=(s,\sigma) \land s \in {\sf Safe}_i \}.$$
	% % % \item {\em Parity objectives}: the objectives are given by $(p_1,p_2,\dots,p_n)$ where each $p_i : S \rightarrow \{0,1,\dots,d\}$ implicitly defines the set of winning plays $$W_i=\{ \rho \in \plays(G) \mid \min \{ j \mid \inf(\rho)\cap p_i^{-1}(j)\not=\emptyset\} \mbox{~is even~}  \}.$$
  % \end{itemize}

% \paragraph{{\bf Doomsday equilibrium and imperfect information}} 
The notion of doomsday equilibrium is defined as for games with
perfect information but with the additional requirements that {\em only} observation-based strategies can be used by the players. Given an $n$-player game arena with imperfect information $G$
and $n$ winning objectives $(\varphi_i)_{1 \leq i \leq n}$ (defined as in the perfect information setting), we
want to solve the problem of deciding the existence of an {\em observation-based strategy profile} $\Lambda$ which is a doomsday equilibrium in $G$ for $(\varphi_i)_{1 \leq i \leq n}$.

\begin{example} Fig. \ref{fig:example-imperfect} depicts a variant of 
the example in the perfect information setting, with imperfect information. In this example let us describe the
situation for Player $1$. It is symmetric for the other players. Assume that 
when Player $2$ or Player $3$ send their information to Player $1$ (modeled by
a visit to his happy states), Player $1$ cannot distinguish which of Player $2$ or
$3$ has sent the information, e.g. because of the usage of a cryptographic primitive. 
Nevertheless, let us show that there exists doomsday equilibrium. Assume that the three
players agree on the following protocol: Player $1$ and $2$ send their information but
not Player $3$.

Let us show that this sequence witnesses a doomsday equilibrium and argue that
this is the case for Player $1$. From the point of view of Player $1$, if all players
follow this profile of strategies then the outcome is winning for Player $1$. Now, let us consider
two types of deviation. First, 
assume that Player $2$ does not send his information (i.e. does not visit the happy states). In that case Player $1$ will observe the deviation
and can retaliate by not sending his own information. Therefore all the players are losing. Second, assume that Player $2$ does not send his information but Player $3$ does. In this case
it is easy to verify that Player $1$ cannot observe the deviation and so according to his strategy
will continue to send his information. This is not problematic because all the plays that are
compatible with Player $1$'s observations are such that: $(i)$ they are winning for Player $1$ (note that it would be also acceptable that all the sequence are either winning for Player $1$ or losing for all the other players), and $(ii)$ Player $1$ is always in position to retaliate along this sequence of observations. In our solution below these two properties are central and 
will be called {\em doomsday compatible} and {\em good for retaliation}. $\hfill \square$
\end{example}

\paragraph{{\bf Generic Algorithm}} We present a generic algorithm to test the existence
of an observation-based doomsday equilibrium in a game of imperfect information.
To present this solution, we need two additional notions: sequences of observations which are {\em doomsday compatible} and prefixes which are {\em good for retaliation}. These two notions are defined as follows. In a game arena $G=(S,{\cal P}, s_{{\sf init}}, \Sigma, \Delta, (O_i)_{1\leq i\leq n})$ with imperfect information and winning objectives $(\varphi_i)_{1 \leq i \leq n}$,
  \begin{itemize}
  	\item a sequence of observations $\eta \in (O_i  \times (\Sigma \cup \{ \tau \}))^{\omega}$ is {\em doomsday compatible} (for Player~$i$) if $\gamma_i(\eta)\subseteq \varphi_i \cup \bigcap_{j=1}^{j=n} \overline{\varphi_j}$, i.e. all plays that are compatible with $\eta$ are either winning for Player~$i$, or not winning for any other player,
	\item a prefix $\kappa \in (O_i \times (\Sigma \cup \{ \tau \}))^{*} \cdot O_i$ of a sequence of observations is {\em good for retaliation} (for Player~$i$) if there exists an observation-based strategy $\lambda_i^R$ such that for all prefixes $\pi \in \gamma_i(\kappa)$ compatible with $\kappa$, 
          $\outcome(\pi,\lambda^R_i) \subseteq \varphi_i \cup \bigcap_{j=1}^{j=n} \overline{\varphi_j}$. 
  \end{itemize}

The next lemma shows that the notions of sequences of observations that are doomsday compatible and 
good for retaliation prefixes are important for studying the existence of doomsday equilibria for 
imperfect information games.

\begin{lemma}
\label{general-construct}
Let $G$ be an $n$-player game arena with imperfect information and winning objectives $\varphi_i$, $1 \leq i \leq n$. 
There exists a doomsday equilibrium in G if and only if there exists a play $\rho$ in $G$ such that:
  \begin{itemize}
  	\item[$(F_1)$] $\rho \in \bigcap_{i=1}^{i=n} \varphi_i$, i.e. $\rho$ is winning for all the players,
	\item[$(F_2)$] for all Player~$i$, $1 \leq i \leq n$, for all prefixes $\kappa$ of $\Obs_i(\rho)$, $\kappa$ is good for retaliation for Player~$i$,
	\item[$(F_3)$] for all Player~$i$, $1 \leq i \leq n$, $\Obs_i(\rho)$ is doomsday compatible for Player~$i$.
  \end{itemize}
\end{lemma}

\begin{proof}
First, assume that conditions $(F_1), (F_2)$ and $(F_3)$ hold and show that there exists a DE in $G$.
We construct a DE $(\lambda_1, \dots, \lambda_n)$ as follows.
For each player~$i$, the strategy $\lambda_i$ plays according to the (observation of the) 
path $\rho$ in ${\cal G}$, as long as the previous observations follow $\rho$. 
If an observation is unexpected for Player~$i$ (i.e., differs from the sequence
in $\rho$), then $\lambda_i$ switches to an observation-based retaliating strategy $\lambda_i^R$ (we will show that such a strategy exists as a consequence of $(F_2)$).
This is a well-defined profile and a DE because:
$(1)$ all strategies are observation-based, and the outcome of the profile is 
the path $\rho$ that satisfies all objectives;
$(2)$ if no deviation from the observation of $\rho$ is detected by Player~$i$, then by condition $(F_3)$ we know that if the outcome does not satisfy $\varphi_i$, then it does not satisfies $\varphi_j$, for all $1 \leq j \leq n$, 
$(3)$ if a deviation from the observation of $\rho$ is detected by Player~$i$, then
the sequence of observations of Player $i$ so far can be decomposed as
$\kappa = \kappa_1(o_1,\sigma_1)\dots(o_m,\sigma_nm$ where $(o_1,\sigma_1)$ is the first deviation
of the observation of $\rho$, and $(o_m,\sigma_m)$ is the first time it is Player $i$'s turn to play
after this deviation (so possibly $m=1$).  By condition $(F_2)$, we know that $\kappa_1$
is good for retaliation. Clearly, $\kappa_1(o_1,\sigma_1)\dots (o_\ell,\sigma_\ell)$ is retaliation
compatible as well for all $\ell\in\{1,\dots,m\}$ since retaliation goodness is preserved 
by player $j$'s actions for all $j$. Therefore $\kappa$ is good for retaliation and by definition
of retaliation goodness there exists an observation-based
retaliation strategy $\lambda_i^R$ for Player $i$ which ensures that that  regardless of the strategies
of the opponents in coalition, if the outcome does not satisfy $\varphi_i$,
then for all $j\in\{1,\dots,n\}$, it does not satisfy $\varphi_j$ either.
%
%Hence, $(\lambda_1, \dots, \lambda_n)$ is a doomsday equilibrium.

Second, assume that there exists a DE $(\lambda_1, \dots, \lambda_n)$ in $G$,
and show that $(F_1),(F_2)$ and $(F_3)$ hold.
Let $\rho$ be the outcome of the profile $(\lambda_1, \dots, \lambda_n)$. Then
$\rho$ satisfies $(F_1)$ by definition of DE. 
Let us show that it also satisfies $(F_3)$.
By contradiction, if $\obs_i(\rho)$ is not doomsday compatible for Player $i$, then
by definition, there is a path $\rho'$ in $\plays(G)$ that is compatible with the observations
and actions of player~$i$ in $\rho$ (i.e., $\obs_i(\rho) = \obs_i(\rho')$),
but $\rho'$ does not satisfy $\varphi_i$, while it satisfies $\varphi_j$ for 
some $j \neq i$.  Then, given the strategy $\lambda_i$ from the profile, the other players in coalition
can choose actions to construct the path~$\rho'$~(since $\rho$ and $\rho'$ are
observationally equivalent for player~$i$, the observation-based strategy~$\lambda_i$ is going
to play the same actions as in $\rho$). This would show that the profile
is not a DE, establishing a contradiction. Hence $\obs_i(\rho)$ 
is doomsday compatible for Player $i$ for all $i=1,\dots,n$ and $(F_3)$ holds.
Let us show that $\rho$ also satisfies $(F_2)$. Assume that this not true.
Assume that $\kappa$ is a prefix of $\obs_i(\rho)$ such that 
$\kappa$ is not good for retaliation for Player $i$ for some $i$. 
By definition it means that the other players can 
make a coalition and enforce an outcome $\rho'$, from any prefix of play compatible with $\kappa$, that
is winning for one of players of the coalition, say Player~$j$, $j\not=i$, and losing for Player~$i$. 
This contradicts the fact that $\lambda_i$ belongs to a DE.\qed
\end{proof}

\begin{theorem}\label{thm:imperfect}
The problem of deciding the existence of a doomsday equilibrium in an
$n$-player game arena with imperfect information and $n$ objectives 
is {\sc ExpTime-C} for objectives that are either
all reachability, all safety, all B\"uchi, all co-B\"uchi or all parity objectives.
\end{theorem}
\begin{proof}
By Lemma~\ref{general-construct}, we know that we can decide the existence of a doomsday equilibrium by checking the existence of a play $\rho$ in $G$ that respects the conditions $(F_1), (F_2)$, and $(F_3)$.
It can be shown (see Appendix), for all $i\in\{1,\dots,n\}$, 
that the set of good for retaliation prefixes for Player $i$ is definable by a finite-state automaton
$C_i$, and the set of observation sequences that are doomsday compatible for Player $i$ is definable by
a Streett automaton $D_i$. 

From the automata $(D_i)_{1 \leq i \leq n}$ and $(C_i)_{1 \leq i \leq n}$, we construct using a synchronized product a finite transition system $T$ and check for the existence of a path in $T$ that satisfy the winning objectives for each player in $G$, the Streett acceptance conditions of the $(D_i)_{1 \leq i \leq n}$, and whose all prefixes are accepted by the automata $(C_i)_{1 \leq i \leq n}$. The size of $T$ is exponential in $G$ and the acceptance condition is a conjunction of Streett and safety objectives. The existence of such a path can be established in polynomial time in the size of $T$, so in exponential time in the size of $G$. The {\sc ExpTime}-hardness 
is a consequence of the {\sc ExpTime}-hardness of two-player games of imperfect information for all the considered objectives \cite{DBLP:conf/fsttcs/BerwangerD08,RaskinCDH07}.\qed
\end{proof}
%%% Local Variables: 
%%% mode: latex
%%% TeX-master: t
%%% End: 

\section{Conclusion}

We defined the notion of doomsday threatening equilibria both for perfect and imperfect information $n$ player games with omega-regular objectives. This notion generalizes 
to $n$ player games the winning secure equilibria \cite{CHJ06}. Applications in the analysis of security protocols are envisioned and will be pursued as future works.

We have settled the exact complexity in games of perfect information for almost all omega-regular objectives with complexities ranging from {\sc PTime} to {\sc PSpace}, the only small gap that remains is for parity objectives where we have a {\sc PSpace} algorithm and both {\sc NP} and {\sc coNP}-hardness. Surprisingly, the existence of doomsday threatening equilibria in $n$ player games with imperfect information is decidable and more precisely {\sc ExpTime-C} for all the objectives that we have considered.

In a long version of this paper~\cite{longversion}, we provide a solution in {\sc 2ExpTime} for deciding the existence of a doomsday threatening equilibrium in a game whose objectives are given as LTL formula (this solution is optimal as it is easy to show that the classical LTL realizability problem can be reduced to the DE existence problem). We also provide a Safraless solution~\cite{KupVar05a} suitable to efficient implementation. 
\vspace{-3mm}

%%% Local Variables: 
%%% mode: latex
%%% TeX-master: t
%%% End: 

% \input{examples}

% \bibliographystyle{plain}
% \bibliography{bib1}

\makeatletter
\newcounter{suitebib}
\let\@save@bibitem\@bibitem
\def\@bibitem#1{\@save@bibitem{#1}
  \global\setcounter{suitebib}{\value{\@listctr}}}
\makeatother
\putbib[biblio]
\end{bibunit}

\newpage
\begin{bibunit}[abbrv]

\appendix
\section{Complexity of Doomsday Equilibria for Perfect Information Games}

\subsection{Proof of Lemma \ref{lem:lowerboundstail}}

The proof of Lemma \ref{lem:lowerboundstail} is decomposed into several lemmas.

\begin{lemma}
The problem of deciding the existence of a doomsday equilibrium in a $2$-player game arena is {\sc PTime}-hard both for B\"uchi and co-B\"uchi winning objectives.
\end{lemma}
\begin{proof}
We explain the result for B\"uchi objectives (the proof for co-B\"uchi objectives is similar). 
To establish this result, we show how to reduce the problem of deciding the winner in a two-player zero-sum game with a B\"uchi objective (for Player~1), a {\sc PTime-C} problem \cite{Imm4}, can be reduced to the existence of a doomsday equilibrium in a two-player game arena with B\"uchi objectives. 
Let $G$ be the two-player game, $S$ its set of states, and $T$ the set of states that Player~1 wants to visit infinitely often. We reduce the problem of deciding the existence of such a strategy to the existence of a doomsday equilibrium in the same game arena,
where the objective of Player~1 is the original B\"uchi objective, i.e. ${\sf B\ddot uchi}(T)$, and the objective of Player~2 is trivial: ${\sf B\ddot uchi}(S)$. Clearly, as Player~2 will
always satisfy his objective, Player~1 must have a winning strategy for ${\sf B\ddot uchi}(T)$ if a doomsday equilibrium exists (and vice versa) otherwise condition~2 would be violated.   \qed
\end{proof}

We now turn to the proof of lower bounds for parity objectives. We build our proof for these lower bounds on the hardness of generalized parity games~\cite{ChatterjeeHP07}: in a two-player (called Player~$A$ and Player~$B$) zero-sum game,
 %with perfect information
and an objective given by:

  \begin{itemize}
  	\item the {\em disjunction} of two parity objectives, it is {\sc NP-Hard} to decide if Player~$A$ has a winning strategy, 
	\item the {\em conjunction} of two parity objectives, it is {\sc coNP-Hard} to decide if Player~$A$ has a winning strategy.
  \end{itemize}
\noindent
We next show that these decision problems can be reduced to the problem of the existence of a doomsday equilibrium in an $n$ player game with parity objectives. 

\begin{figure*}[!ht]
\centering

\begin{tikzpicture}[scale=0.8]
\tikzstyle{every state}=[fill=gray!20!white]

\tikzstyle{sttstates} = [rounded corners, fill=gray!10!white,fill opacity=0.5]

\tikzset{P3/.style={fill=gray!20!white,diamond,draw,very thick,minimum size=6mm}}
\tikzset{P2/.style={fill=gray!20!white,rectangle,draw,very thick,minimum size=10mm}}
\tikzset{P1/.style={fill=gray!20!white,circle,draw,very thick,minimum size=10mm}}

\draw[line width=.5mm,color=black] (-3,4) node[P3](sinit'){$s_{init}'$};

\draw[line width=.5mm,color=black] (0,4) node[P1](sinit){$s_{init}$};

\draw[line width=.5mm,color=black] (3,4) node[P2](s1){$s_1$};

\draw[line width=.5mm,color=black] (6,4) node[P3](s2){$(s_2,3)$};

\draw[line width=.5mm,color=black] (9,4) node[P1](s3){$s_2$};

\draw[line width=.5mm,color=black] (0,7) node[state](sink){${\sf Bad}_1$};

\draw[line width=.5mm,color=black] (4.5,7) node[state](sink'){${\sf Bad}_{1,2}$};

\path[rounded corners, draw=black!50, dashed]
            (-1.2,2) rectangle (9.7,6); 

\draw (5,2.2) node[](label){Modified Copy of $G$};

\path (sinit') edge [loop above] node {$\sigma\in\Sigma$} (sinit');

\path (sinit') edge [->, above] node {${\sf deviate}$} (sinit);

% \path (s1) edge [loop above] node {$(\textsf{stay},\sigma_2^k)$} (s1);

\path (sinit) edge [->,above] node {$\sigma\in\Sigma$} (s1);

\path (s1) edge [->,above] node {$\sigma\in\Sigma$} (s2);

\path (s2) edge [->,above] node {$\sigma\in\Sigma$} (s3);

\path (sinit) edge [->,left] node {${\sf deviate}$} (sink);

\path (sink) edge [loop left] node {$\Sigma'$} (sink);

\path (s1) edge [->,left] node {${\sf deviate}$} (sink');

\path (s2) edge [->,right] node {${\sf deviate}$} (sink');

\path (sink') edge [loop right] node {$\Sigma'$} (sink');

\draw[line width=.5mm,color=black] (-3,3) node[sttstates] {$(0,0,0)$};
\draw[line width=.5mm,color=black] (1.6,2.3) node[sttstates] {$(1, p_1+1,\ p_2+1)$};
\draw[line width=.5mm,color=black] (1.3,6.9) node[sttstates] {$(1,0,0)$};
\draw[line width=.5mm,color=black] (3.1,6.9) node[sttstates] {$(1,1,1)$};

\end{tikzpicture}
  \caption{Structure of the reduction from generalized parity game with a conjunction of two parity objectives to the existence of a doomsday equilibrium with parity objectives.}
  \label{fig:reduc-conj-parity}
\end{figure*}
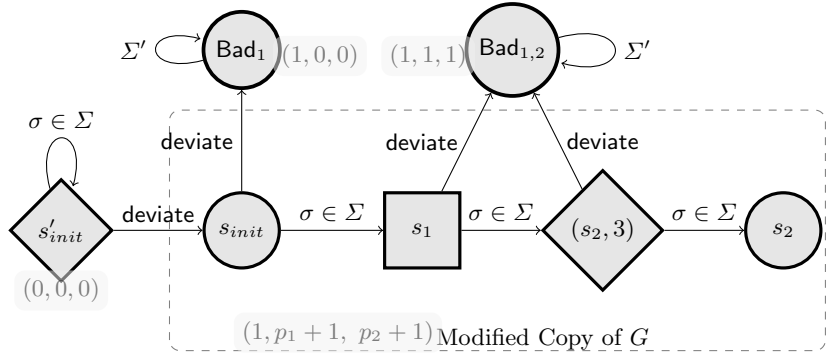

\begin{lemma}[{\sc coNP-Hardness}]
The problem of deciding the existence of a doomsday equilibrium in a $3$-player game arena with parity objectives
is {\sc coNP-hard}.
\end{lemma}
\begin{proof}

Let $G=(S, \{S_A,S_B\}, s_{{\sf init}},\Sigma, \Delta)$ be a two-player game  
and a conjunction of two parity objectives defined by the functions $p_1$ and $p_2$ that Player~$A$ wants to enforce, i.e. the objective of Player~$A$ is to ensure an outcome that satisfies the two parity objectives, while the objective of Player~$B$ is to ensure an outcome that violates at least one of the two parity objectives. W.l.o.g., we assume that $s_{{\sf init}} \in S_A$ and the turns of $A$ and $B$
alternate. 

From $G$, we construct a 3-player game arena $G'=(S', \{S'_1,S'_2,S'_3\},s'_{{\sf init}},\Sigma', \Delta')$ (depicted in Fig. \ref{fig:reduc-conj-parity}), with:
  \begin{itemize}
	\item the set of states $S'=\{ s'_{{\sf init}}, {\sf Bad}_1, {\sf Bad}_{1,2,3} \} \cup S_A \cup S_B  \cup (S_A \times \{3\})$, this set is partitioned as follows: $S_1=S_A \cup \{{\sf Bad}_1\}$, $S_2=S_B$, $S_3=(S_A \times \{3\}) \cup \{s'_{{\sf init}},{\sf Bad}_{1,2,3}\}$.
	\item the initial state is $s'_{{\sf init}}$,
	\item the alphabet of actions is $\Sigma'=\Sigma \cup \{ {\sf deviate} \}$,
   	\item and the transitions of the game $G'$ are defined as follows: 
		\begin{itemize}
			\item For the state $s'_{{\sf init}}$, for all $\sigma \in \Sigma$, $\Delta'(s'_{{\sf init}},\sigma)=s'_{{\sf init}}$, and $\Delta'(s'_{{\sf init}},\sigma)=s_{{\sf init}}$; i.e., the play stay in $s'_{{\sf init}}$, unless Player~3 plays {\sf deviate} in which case the play goes to $s_{{\sf init}}$ that is the copy of the initial state of the game arena $G$.
			\item For all states $s \in S_A$, for all $\sigma \in \Sigma$, $\Delta'(s,\sigma)=\Delta(s,\sigma)$, and $\Delta'(s,{\sf deviate})={\sf Bad}_1$, so the transition function on the copy of $G$ behaves from states owned by Player~1 as in the original game and it sends the game to ${\sf Bad}_1$ if Player~1 plays the action {\sf deviate}.
			\item For all states $s \in S_B$, for all $\sigma \in \Sigma$, $\Delta'(s,\sigma)=(\Delta(s,\sigma),3)$ and $\Delta'(s,{\sf deviate})={\sf Bad}_{1,2,3}$; i.e., if Player~2 plays an action from the game $G$, the effect is to send the game to the Player~3 copy of the same state as in the original game, if he deviates, the game reaches the sink state ${\sf Bad}_{1,2,3}$.
			\item For all states $s \in S_A \times \{3\}$, for all $\sigma \in \Sigma$, $\Delta'((s,3),\sigma)=s$ and $\Delta'((s,3),{\sf deviate})={\sf Bad}_{1,2,3}$. I.e. if Player~3 plays an action $\sigma \in \Sigma$, he gives back the turn to Player~1, otherwise he sends the game to ${\sf Bad}_{1,2,3}$.
			\item The states ${\sf Bad}_1$ and  ${\sf Bad}_{1,2,3}$ are absorbing. 
		\end{itemize}

%			\item the game loops in $s_{{\sf init}}$ as long as the all the players propose to play the action ${\sf \#}$. As soon as one of the three players deviates, then the game moves to the copy $G$, i.e. to its initial state $s_{{\sf init}}$. 
%			\item in all states $q$ of the copy of $G$, as long as Player~$1$ plays an action $\sigma_A \in \Sigma$, and Player~$2$ and Player~$3$ synchronize on a common action $\sigma_B \in \Sigma$, then the next state is defined as in the game $G$: $q'=\Delta'(q,(\sigma_A,\sigma_B,\sigma_B))=\Delta(q,(\sigma_A,\sigma_B))$.
%			\item in all states $q$ of the copy of $G$, if Player~$1$ plays an action $\sigma_A \in \Sigma$ and the two other players do not play a common action or if they both play the action $\#$ then the game evolves to the state $s'_{{\sf sink}}$ (which is never left),
%			\item in all states $q$ of the copy of $G$, if Player~$1$ plays the action $\#$, the game reaches the state $s_{{\sf sink}}$ (which is never left). 
%		\end{itemize}

	\item The parity functions $(p'_i)_{i=1,2,3}$ for the three players are defined to satisfy the following condition:
		\begin{itemize}
			\item first, $p'_i(s'_{{\sf init}})$ is even for all $i=1,2,3$ (so if the game stays there for ever, all the players satisfy their objectives). 
			\item second, in ${\sf Bad}_{1}$ the parity functions return an even number for Player~$2$ and Player~$3$ but an odd number for Player~$1$, this ensures that Player~$1$ should never play the action ${\sf deviate}$ when the game is in the copy of $G$, 
			\item third, in ${\sf Bad}_{1,2,3}$ the parity functions are odd for all the Players. So whenever Player~2 and 3 play ${\sf deviate}$ all players loose,
			\item finally, in the copy of $G$, the parity function is always odd for Player~$1$, and
                          for all states $q\in S_A\cup S_B\cup (S_A\times \{3\})$, 
                          $p'_2(q)=p_1(s)+1$ and $p'_3(q)=p_2(s)+1$, 
                          where $s = q$ if $q\in S_A\cup S_B$, and $s$ is such that $q = (s,3)$ if $q\in S_A\times \{3\}$. 
		\end{itemize} 
  \end{itemize}
\noindent
This concludes the reduction. 

Clearly, since Player $A$ and $B$ always alternate their moves, in the copy of $G$, any play
will eventually reach a state of Player~2 and a state of Player~3, so that they are always
 able to retaliate by playing the action {\sf deviate}.

A doomsday equilibrium exists in $G'$ iff Player~1 is also able to retaliate when the game enter the copy of $G$. 
But clearly, it is possible if and only if he has a strategy to ensure $\overline{{\sf parity}(p'_1)}$ and
$\overline{{\sf parity}(p'_2)}$, or equivalently iff he has a strategy to ensure ${\sf parity}(p_1)$ and 
${\sf parity}(p_2)$, iff Player~$A$ has a winning strategy in the game $G$ for the conjunction of parity objectives
$p_1$ and $p_2$.\qed

 \end{proof}

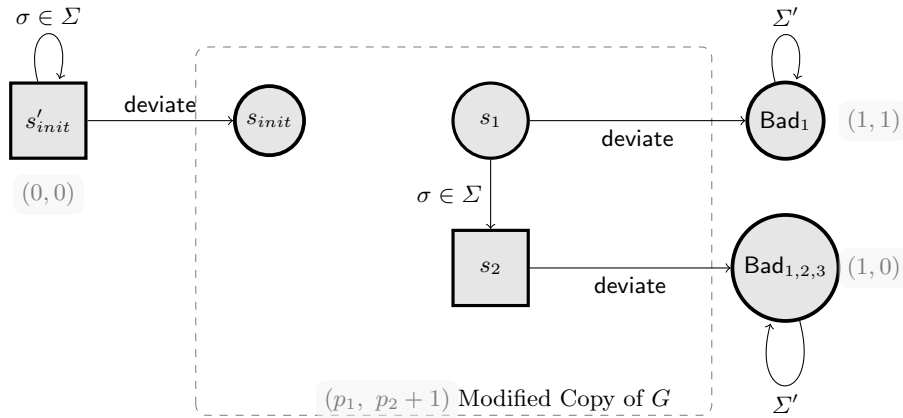
\begin{figure*}[!ht]
\centering

\begin{tikzpicture}[scale=0.98]
\tikzstyle{every state}=[fill=gray!20!white]
\tikzstyle{sttstates} = [rounded corners, fill=gray!10!white,fill opacity=0.5]

\tikzset{P3/.style={fill=gray!20!white,diamond,draw,very thick,minimum size=6mm}}
\tikzset{P2/.style={fill=gray!20!white,rectangle,draw,very thick,minimum size=10mm}}
\tikzset{P1/.style={fill=gray!20!white,circle,draw,very thick,minimum size=10mm}}
\tikzset{P1init/.style={fill=gray!20!white,initial,circle,draw,very thick,minimum size=6mm}}
\tikzset{P1final/.style={fill=gray!20!white,accepting,circle,draw,very thick,minimum size=6mm}}

\draw[line width=.5mm,color=black] (-3,4) node[P2](sinit'){$s_{init}'$};

\draw[line width=.5mm,color=black] (0,4) node[state](sinit){$s_{init}$};

\draw[line width=.5mm,color=black] (3,4) node[P1](s1){$s_1$};

\draw[line width=.5mm,color=black] (3,2) node[P2](s2){$s_2$};

\draw[line width=.5mm,color=black] (7,4) node[state](sink){${{\sf Bad}_1}$};

\draw[line width=.5mm,color=black] (7,2) node[state](sink'){${\sf Bad}_{1,2,3}$};

\path[rounded corners, draw=black!50, dashed]
            (-1,0) rectangle (6,5); 

\draw (4,0.2) node[](label){Modified Copy of $G$};

\path (sinit') edge [loop above] node {$\sigma\in\Sigma$} (sinit');

\path (sinit') edge [->, above] node {${\sf deviate}$} (sinit);

% \path (s1) edge [loop above] node {$(\textsf{stay},\sigma_2^k)$} (s1);

\path (s1) edge [->,left] node {$\sigma\in\Sigma$} (s2);

\path (s1) edge [->,below] node {${\sf deviate}$} (sink);

\path (sink) edge [loop above] node {$\Sigma'$} (sink);

\path (s2) edge [->,below] node {${\sf deviate}$} (sink');

\path (sink') edge [loop below] node {$\Sigma'$} (sink');

\draw[line width=.5mm,color=black] (-3,3) node[sttstates] {$(0,0)$};
\draw[line width=.5mm,color=black] (1.6,0.2) node[sttstates] {$(p_1,\ p_2+1)$};
\draw[line width=.5mm,color=black] (8.2,4) node[sttstates] {$(1,1)$};
\draw[line width=.5mm,color=black] (8.2,2) node[sttstates] {$(1,0)$};

\end{tikzpicture}
  \caption{Structure of the reduction from generalized parity game with a disjunction of two parity objectives to doomsday equilibrium with parity objectives.}
  \label{fig:reduc-disj-parity}
\end{figure*}

\begin{lemma}[{\sc NP-hardness}]
The problem of deciding the existence of a doomsday equilibrium in a $2$-player game arena with  parity objectives is {\sc NP-Hard}.
\end{lemma}
\begin{proof}
For this part, we need to show how to reduce the problem of deciding if Player~$A$ has a winning strategy in a two-player zero-sum game whose objective is defined by the 
disjunction of two parity objectives. We only sketch the construction as it is based on the main ideas used in the {\sc coNP-hardness} result. Let $G=(S, s_{{\sf init}}, \Sigma, \Delta)$ be a two-player game 
with a {\em disjunction} of two parity objectives defined by the functions $p_1$ and $p_2$. The objective of Player~$A$ is to ensure an outcome that satisfies at least one of the two parity objectives (while the objective of Player~$B$ is to ensure an outcome that violates both parity objectives.)

From $G$, we construct a two-player game $G'$ with parity objectives $(p'_i)_{i=1,2}$ (see Fig. \ref{fig:reduc-disj-parity}). The game arena $G'$ contains a copy of $G$ plus three states $s'_{{\sf init}}$ (the initial state), ${\sf Bad}_{1}$ and ${\sf Bad}_{1,2}$. The alphabet of actions is $\Sigma \cup \{ {\sf deviate} \}$. 
The partition of the state space is as follows: $S_1=S_A$ and $S_2=S_B \cup \{ s'_{{\sf init}}\} \cup  \{ {\sf Bad}_{1},{\sf Bad}_{1,2}\}$.
The transitions are as follows: if Player~2 plays $\sigma \in \Sigma$ in $s'_{{\sf init}}$ then the game stays there, if he plays ${\sf deviate}$ then the game enters the copy of $G$. There the transition function for $\sigma \in \Sigma$ is defined as in $G$, and if Player~1 plays ${\sf deviate}$ then the game goes to ${\sf Bad}_1$, and if Player~2 plays ${\sf deviate}$ then the games goes to ${\sf Bad}_{1,2}$.  
The parity functions $(p'_i)_{i=1,2}$ are defined as follows: $p'_1$ returns an even number in $s'_{{\sf init}}$, is equal to $p_1$ in the copy of $G$, returns an odd number in ${\sf Bad}_{1}$ and ${\sf Bad}_{1,2}$. The function $p'_2$ returns an even number in $s'_{{\sf init}}$, is equal to $p_2+1$ in the copy of $G$ (so $p'_2$ is the complement of $p_2$), returns an even number in ${\sf Bad}_1$ and an odd number in ${\sf Bad}_{1,2}$. This definition of $p'_1$ and $p'_2$ ensures that the two players meet their parity objectives when the game always stay in $s'_{{\sf init}}$ and when the game enters the copy of $G$, Player~2 can always retaliate while Player~1 can retaliate if and only if Player~$A$ has a winning strategy in the original game. 
\qed
 \end{proof}

\subsection{Proof of Lemma \ref{lem:hardness-safety-perfect}}

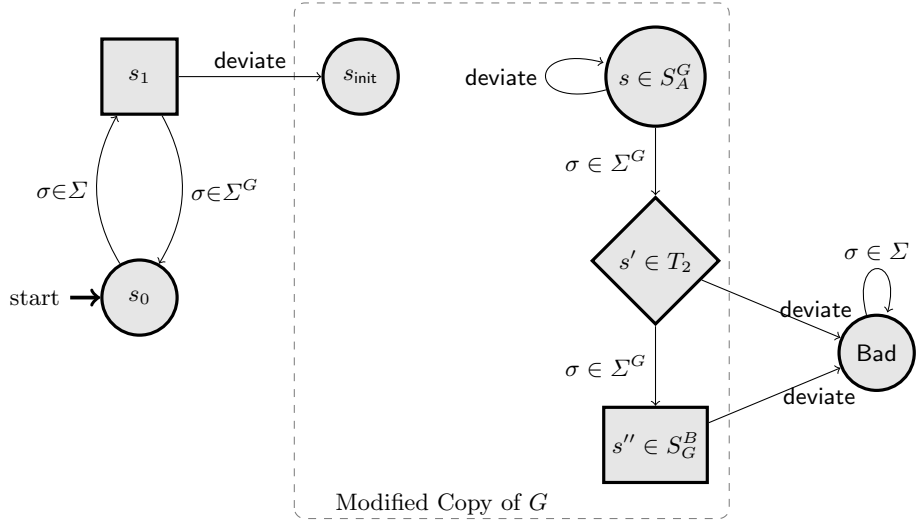
\begin{figure*}[!ht]
\centering

\begin{tikzpicture}[scale=0.98]
\tikzstyle{every state}=[fill=gray!20!white]
\tikzstyle{sttstates} = [rounded corners, fill=gray!10!white,fill opacity=0.5]

\tikzset{P3/.style={fill=gray!20!white,diamond,draw,very thick,minimum size=6mm}}
\tikzset{P2/.style={fill=gray!20!white,rectangle,draw,very thick,minimum size=10mm}}
\tikzset{P1/.style={fill=gray!20!white,circle,draw,very thick,minimum size=10mm}}
\tikzset{P1init/.style={fill=gray!20!white,initial,circle,draw,very thick,minimum size=10mm}}
\tikzset{P1final/.style={fill=gray!20!white,accepting,circle,draw,very thick,minimum size=6mm}}

\draw[line width=.5mm,color=black] (-3,2) node[P1init](s0){$s_0$};

\draw[line width=.5mm,color=black] (-3,5) node[P2](s1){$s_1$};

\draw[line width=.5mm,color=black] (0,5) node[P1](sinit){$s_{\sf init}$};

\draw[line width=.5mm,color=black] (4,5) node[P1](s){$s\in S_A^G$};

\draw[line width=.5mm,color=black] (4,2.5) node[P3](s'){$s'\in T_2$};

\draw[line width=.5mm,color=black] (4,0) node[P2](s''){$s''\in S_G^B$};

\draw[line width=.5mm,color=black] (7,1.25) node[P1](bad){${\sf Bad}$};

\path[rounded corners, draw=black!50, dashed]
            (-0.9,-1) rectangle (5,6); 

\draw (1.1,-0.8) node[](label){Modified Copy of $G$};

\path (s1) edge [->,bend left,right] node {$\sigma{\in}\Sigma^G$} (s0);

\path (s1) edge [->, above] node {${\sf deviate}$} (sinit);

\path (s0) edge [->,bend left, left] node {$\sigma{\in}\Sigma$} (s1);

\path (s) edge [->,left] node {$\sigma\in\Sigma^G$} (s');

\path (s') edge [->,left] node {$\sigma\in\Sigma^G$} (s'');

% \path (s1) edge [->,below] node {$\begin{array}{l}(\_,\overline{\sigma}) \text{ where } \overline{\sigma}\neq \sigma_2^k\\ \text{for any }\sigma_2\in \Sigma\end{array}$} (bad);

\path (s') edge [->, right] node {${\sf deviate}$} (bad);

\path (s'') edge [->, right] node {${\sf deviate}$} (bad);

\path (s) edge [loop left] node {${\sf deviate}$} (s);

\path (bad) edge [loop above] node {$\sigma\in\Sigma$} (bad);

% \draw[line width=.5mm,color=black] (-1.3,3.3) node[sttstates] {$1$};
% \draw[line width=.5mm,color=black] (0.3,1.3) node[sttstates] {$2$};
% \draw[line width=.5mm,color=black] (9,2.6) node[sttstates] {$3$};

\end{tikzpicture}
  \caption{Structure of the reduction from multi-reachability game to doomsday equilibrium with safety objectives. Round nodes
    denote Player $0$'s states and rectangular nodes denote Player $1$'s states. }
  \label{fig:reduc-multi-reach}
\end{figure*}

\begin{proof}
We reduce the two-player {\em multi-reachability} problem to our problem, {\sc PSpace-Hardness} follows.
Let $G=(S^G, \{S^G_A,S^G_B\}, s^G_{{\sf init}}, \Sigma^G, \Delta^G)$ be a two-player (Player~$A$ and Player~$B$) game arena. Let ${\cal T}=\{ T_1,T_2,\dots,T_k \}$ be a family of 
subsets of $S^G$ supposed to be pairwise disjoint (w.l.o.g.). Also wlog we assume that $T_i\subseteq S^G_B$ for all $i\in \{1,\dots,k\}$. 
 In a multi-reachability game, the objective of Player~$A$ is to visit each $T$ in ${\cal T}$, while Player~$B$ tries to avoid at least one of the subsets in ${\cal T}$. So, multi-reachability games are two-player zero sum games where the winning plays for Player~$A$ are 
$$\{ \rho=s_0 s_1 \dots s_n \dots \in \plays(G) \mid \forall i \cdot 1 \leq i \leq k \cdot \exists j \geq 0 \cdot s_j \in T_i \}.$$ 
It has been shown that the multi-reachability problem for two-player games is {\sc PSpace-C}~\cite{AlurTM03,Fijalknov10}. 

From $G=(S^G, \{S^G_A,S^G_B\}, s^G_{{\sf init}}, \Sigma^G, \Delta^G)$ and ${\cal T}=\{ T_1,T_2,\dots,T_k \}$ that define a multi-reachability game, we construct a game arena $G'=(S,\{ S_0,S_1,\dots,S_k\},\\$ $s_0,\Sigma,\Delta)$ with $k+1$ players and a set of $k+1$ safety objectives $\safe_0,\dots,\safe_{k}$ such that Player~$A$ wins the multi-reachability objective defined by $G$ and ${\cal T}$ iff there exists a doomsday equilibrium in $G'$ for the safety objectives $\safe_0,\dots,\safe_{k}$.  

The structure of the reduction is depicted in Fig.~\ref{fig:reduc-multi-reach}. The state space of $G'$ is composed of three parts: an initial part on the left, a modified copy of $G$, and a part on the right.
The set $S$ of states is  $\{s_0,s_1\} \cup S_A^G \cup S_B^G \cup \{ {\sf Bad} \}$. This set of states is partitioned as follows: $S_0 = \{s_0\} \cup S^G_A \cup \{ {\sf Bad} \}$, 
$S_1=S_B^G \setminus \bigcup_{i=2}^{i=k} T_i$ and for all $i$, $2 \leq i \leq k$, $S_i=T_i$.

\noindent
The sets of safety objectives are defined as follows: $\safe_{0}={\sf safe}(\{s_0,s_1\})$, and for all $i \in \{1,2,\dots,k\}$, $\safe_i={\sf safe}(S \setminus (\{{\sf Bad}\} \cup T_i))$.
The alphabet of actions is $\Sigma=\Sigma^G \cup \{ {\sf deviate} \}$, and the transition function is defined as follows:
  \begin{itemize}
  	\item $\Delta(s_0,\sigma)=s_1$, for all $\sigma \in \Sigma$,
	\item $\Delta(s_1,\sigma)=\left \{ \begin{array}{ll}
							s_0 & \mbox{if~} \sigma\in \Sigma^G \\
							s^G_{{\sf init}} & \mbox{if~} \sigma={\sf deviate}
						 \end{array} \right.$
	\item for all $s \in S_A^G \cup S_B^G$:
		\begin{itemize}
			\item for all $\sigma \in \Sigma^G$, $\Delta(s,\sigma)=\Delta^G(s,\sigma)$
			\item for the letter {\sf deviate}: for all states $s \in S_A^G$, $\Delta(s,{\sf deviate})=s$, and for all $s \in S_B^G$, $\Delta(s,{\sf deviate})={\sf Bad}$
		\end{itemize}
              \item ${\sf Bad}$ is a sink state.
  \end{itemize} 
\noindent
Now, let us justify this construction. First, assume that, in the two-player game arena $G=(S^G, \{S^G_1,S^G_2\}, s^G_{{\sf init}}, \Sigma^G, \Delta^G)$ with the multi-reachability
objective given by ${\cal T}=\{ T_1,T_2,\dots,T_k \}$, Player~$A$ has a winning strategy. In that case, we show that there exists a doomsday equilibrium in the game $G'$ for the safety objectives $(\safe_i)_{0 \leq i \leq k}$. 
To establish the existence of a doomsday equilibrium, we consider the strategy profile $\Lambda=(\lambda_0,\lambda_1,\dots,\lambda_{k})$ whose strategies respect the following conditions:
  \begin{itemize}
  	\item If all the players follows the strategy profile $\Lambda$, the outcome of the game is $(s_0 \cdot s_1)^{\omega}$, i.e. Player~1 avoids to play ${\sf deviate}$ in $s_1$.
	\item Whenever player~1 plays {\sf deviate} in $s_1$, then the game enters the sub game of $G'$ corresponding to $G$, and the game thus enters an unsafe state for Player~$0$ (as $s_{{\sf init}}$ is not part of $\safe_{0}$). From there, Player~$0$ must retaliate by forcing a visit to each set in ${\cal T}=\{ T_1,T_2,\dots,T_k \}$ to make sure that all the other players lose. By hypothesis, in $G$, Player~$A$ has a winning strategy for the multi-reachability objective, so we know that if the other players play letters that are in $\Sigma^G$ then all sets in ${\cal T}$ will eventually be visited when Player~$0$ plays according to the winning strategy of Player~$A$ in $G$. On the other hand, if the letter ${\sf deviate}$ is played then the game goes to the state ${\sf Bad}$ where all the safety objectives are violated. So, we have established that Player~$0$ can retaliate if he plays as Player~$A$ in the copy of $G$. Now, let us consider all the other players. According to the definition of the transition function, Player~$i$ has the option to retaliate whenever he enters its unsafe set $T_i$ by choosing the action ${\sf deviate}$ and so force a visit to ${\sf Bad}$. So, all other players have also the ability to retaliate whenever they enter their unsafe region.	  
  \end{itemize}
 \noindent
So, we have established that $(\lambda_0,\lambda_1,\dots,\lambda_{k})$ witnesses a doomsday equilibrium in $G'$.

Now, let us consider the other direction. Let $(\lambda_0,\lambda_1,\dots,\lambda_{k})$ be a profile of strategies which witnesses a doomsday equilibrium for $G'$ and the safety objectives given by 
the subsets of plays $(\safe_i)_{i=0,..,k}$. In that case, if we consider a prefix of play that enters for the first time the state $s_{{\sf init}}$, we know by definition of doomsday equilibrium that Player~$0$ has a strategy to retaliate against any strategies of the adversaries. If all the other players chooses their letters in $\Sigma^G$ then it should be the case that the play visits all the sets in ${\cal T}$. So, this clearly means that Player~$A$ has a winning strategy in $G$ for the multi-reachability objective defined by ${\cal T}$, this strategy simply follows the strategy $\lambda_{0}$ in the copy of $G$.
\qed
\end{proof}

\section{Complexity of Doomsday Equilibria for Imperfect Information Games}

We now present automata construction to recognize sequences of observations that are doomsday compatible and prefixes that are good for retaliation.

\begin{lemma}
\label{good-retaliation}
Given an $n$-player game $G$ with imperfect information and a set of reachability, safety or parity objectives $(\varphi_i)_{1 \leq i \leq n}$, we can construct for each Player~$i$, in exponential time,  a deterministic Streett
automaton $D_i$ whose language is exactly the set of sequences of observations $\eta \in (O_i \times (\Sigma \cup \{ \tau \}))^{\omega}$ that are {\em doomsday compatible} for Player~$i$, i.e. 
$$
L(D_i)=\{ \eta \in (O_i \times (\Sigma \cup \{ \tau\} ))^{\omega} \mid \forall \rho \in \gamma_i(\eta) \cdot \rho \in \varphi_i \cup \bigcap_{j\not=i} \overline{\varphi_j} \}.
$$ 
\noindent
For each $D_i$, the size of its set of states is bounded by ${\bf O}(2^{nk \log k})$ and the number of Streett pairs is bounded by ${\bf O}(nk^2)$ where $k$ is the number of states in $G$. 
\end{lemma}
\begin{proof}
Let $G=(S,(S_i)_{1\leq i\leq n}, s_{{\sf init}}, \Sigma, \Delta, (O_i)_{1\leq i\leq n})$, and let us show the constructions for Player~$i$, $1 \leq i \leq n$.
We treat the three types of winning conditions as follows.

We start with safety objectives. Assume that the safety objectives are defined implicitly by the following tuple of sets of safe states: $(T_1,T_2,\dots,T_n)$, i.e. $\varphi_i = {\sf safe}(T_i)$. First, we construct the automaton 
$$A=(Q^A,q^A_{{\sf init}},(O_i \times (\Sigma \cup \{ \tau \}),\delta^A)$$ over the alphabet $O_i \times (\Sigma \cup \{ \tau \})$ as follows:
  \begin{itemize}
  	\item $Q^A=S$, i.e. the states of $A$ are the states of the game structure $G$,
	\item $q^A_{{\sf init}}=s_{{\sf init}}$,
	\item $(q,(o,\sigma),q') \in \delta^A$ if $q\in o$ and there exists $\sigma'\in\Sigma$ such that $\Delta(q, \sigma')=q'$ and such that $\sigma = \tau$ if $q\not\in S_i$, and $\sigma = \sigma'$ if $q\in S_i$.

          % either $(i)$ $\sigma \not= \tau$, $q \in S_i$, $q \in o$, and $\Delta(\sigma,q)=q'$, or $(ii)$ $\sigma=\tau$, $q \not\in S_i$, $q \in o$, and $\Delta(q,\sigma)=q'$,
	% \item $F = T_i$ and $B^A_j =S\setminus T_j$, for all $j\not=i$.
  \end{itemize}
\noindent
The acceptance condition of $A$ is \emph{universal} and expressed with LTL syntax: 
\begin{center}
A word $w$ is accepted by $A$ iff \emph{all} runs $\rho$ of $A$ on $w$ satisfy $\rho \models \Box T_i \lor \bigwedge_{j\not=i} \Diamond \overline{T_j}$.
\end{center}
\noindent
 Clearly, the language defined by $A$ is exactly the set of sequences of observations $\eta \in (O_i\times (\Sigma \cup \{ \tau \}))^{\omega}$ that are {\em doomsday compatible} for Player~$i$, this is because the automaton $A$ checks (using universal nondeterminism) that all plays that are compatible with a sequence of observations are doomsday compatible.

Let us show that we can construct a deterministic Streett automaton $D_i$ that accepts the language of $A$ and whose size is such that: $(i)$ its number of states is at most ${\bf O}(2^{(nk \log k)})$ and $(ii)$ its number of Streett pairs is at most ${\bf O}(nk)$. We obtain $D$ with the following series of constructions:
  \begin{itemize}
  	\item First, note that we can equivalently see $A$ as the intersection of the languages of $n-1$ universal automata $A_j$ with the acceptance condition $\Box T_i \lor \Diamond \overline{T_j}$, $j \not=i$, $1 \leq j \leq n$.
	\item Each $A_j$ can be modified so that a violation of $T_i$ is made permanent and a visit to $\overline{T_j}$ is recorded. 
          For this, we use a state space which is equal to $Q^A \times \{0,1\} \times \{0,1\}$, the first bit records a visit to $\overline{T_i}$
          and the second a visit to $\overline{T_j}$. We denote this automaton by $A'_j$, and its acceptance condition is now
          $\Box \Diamond (Q^A\times \{0,1\}\times \{0\}) \rightarrow \Box \Diamond (Q^A\times \{0\}\times \{0,1\})$. Clearly, this is a universal Streett automaton with a single Streett pair.
	\item $A'_j$, which is a universal Streett automaton, can be complemented (by duality) by interpreting it as a nondeterministic Rabin automaton (with one Rabin pair). This nondeterministic Rabin automaton can be made deterministic using a Safra like procedure, and according to~\cite{CZL09} we obtain a deterministic Rabin automaton with ${\bf O}(2^{k \log k})$ states and ${\bf O}(k)$ Rabin pairs. Let us call this automaton $A''_j$.
 	\item Now, $A''_j$ can be complemented by considering its Rabin pairs as Streett pairs (by dualization of the acceptance condition): we obtain a deterministic Streett automaton with ${\bf O}(k)$ Streett pairs for each $A_j$.
	\item Now, we need to take the intersection of the $n-1$ deterministic automata $A''_j$ (interpreted as Streett automata). Using a classical synchronized product we obtain a single deterministic Streett automaton $D_i$ of size with ${\bf O}(2^{nk \log k})$ states and ${\bf O}(nk)$ Streett pairs. This finishes our proof for safety objectives.
\end{itemize}

Let us now consider reachability objectives. Therefore we now assume the states in $T_1,\dots,T_n$ to be target states for each player respectively, i.e. 
$\varphi_i = {\sf reach}(T_i)$. 
The construction is in the same spirit as the construction for safety. 
Let $A=(Q^A,q^A_{{\sf init}},O_i \times (\Sigma \cup \{ \tau \}),\delta^A)$ be the automaton over $(O_i \times (\Sigma \cup \{ \tau \})$
constructed from $G$ as for safety, with the following (universal) acceptance condition;
\begin{center}
 A word $w$ is accepted by $A$ iff all runs $\rho$ of $A$ on $w$ satisfy $\rho \models  ( \bigvee_{j\not=i} \Diamond T_j ) {\rightarrow} \Diamond T_i$.
 \end{center}
 \noindent
 Clearly, the language defined by $A$ is exactly the set of sequences of observations $\eta \in ((\Sigma \cup \{ \tau \}) \times O_i)^{\omega}$ that are {\em doomsday compatible} for Player~$i$ (w.r.t. the reachability objectives). Let us show that we can construct a deterministic Streett automaton $D_i$ that accepts the language of $A$ and whose size is such that: $(i)$ its number of states is at most ${\bf O}(2^{(nk \log k)})$ and $(ii)$ its number of Streett pairs is at most ${\bf O}(nk)$.
We obtain $D_i$ with the following series of constructions:
  \begin{itemize}
  	\item First, the acceptance condition can be rewritten as $\bigwedge_{j\not=i} (\Diamond T_j  \rightarrow \Diamond T_i)$. Then clearly if $A_j$ is a copy of $A$
          with acceptance condition $\Diamond T_j  \rightarrow \Diamond T_i$ then $L(A)=\bigcap_{j\not=i} L(A_j)$.
	\item For each $A_j$, we construct a universal Streett automaton with one Streett pair by memorizing the visits to $T_i$ and $T_j$ and considering the acceptance condition $\Box \Diamond T_j \rightarrow \Box \Diamond T_i$. So, we get a universal automaton with a single Streett pair.
	\item Then we follow exactly the last three steps (3 to 5) of the construction for safety.
  \end{itemize}	

Finally, let us consider parity objectives. The construction is similar to the other cases. Specifically, we can take as acceptance condition for $A$ the universal condition
$\bigwedge_{j \not= i} ({\sf parity}_i \lor \overline{{\sf parity}_j})$, and treat each condition ${\sf parity}_i \lor \overline{{\sf parity}_j}$ separately. 
We dualize the acceptance condition of $A$, into the nondeterministic condition $\overline{{\sf parity}_i} \land {\sf parity}_j$. This acceptance condition can be equivalently expressed as a Streett
condition with at most ${\bf O}(k)$ Streett pairs. This automaton accepts exactly the set of observation sequences that are not doomsday compatible for Player~$i$ against Player~$j$. Now, using optimal procedure for determinization, we can obtain a deterministic Rabin automaton, with ${\sf O}(k^2)$ pairs that accepts the same language \cite{journals/lmcs/Piterman07}. Now, by interpreting the pairs of the acceptance condition as Streett pairs instead of Rabin pairs, we obtain a deterministic Streett automaton $A_j$ that accepts the set of observations sequences that are doomsday compatible for Player~$i$ against Player~$j$. Now, it suffices to take the product of the $n-1$ deterministic Streett automata $A_j$ to obtain the desired automaton $A$, its size is at most ${\bf O}(2^{nk \log k})$ with at most ${\bf O}(n k^2)$ Streett pairs.  \qed
\end{proof}

\begin{lemma}
\label{retaliation-compatible}
Given an $n$-player game arena $G$ with imperfect information and a set of reachability, safety or parity objectives $(\varphi_i)_{1 \leq i \leq n}$, for each Player~$i$, we can
construct a finite-state automaton $C_i$ that accepts exactly the prefixes of observation sequences that are good for retaliation for Player~$i$. 
\end{lemma}
\begin{proof}
Let us show how to construct this finite-state automaton for any Player~$i$, $1 \leq i \leq n$. Our construction follows these steps:
  \begin{itemize}
  	\item First, we construct from $G$, according to lemma~\ref{good-retaliation}, a deterministic Streett automaton
          $D_i=(Q^{D_i},q^{D_i}_{{\sf init}},(O_i \times (\Sigma \cup \{ \tau \}),\delta^{D_i},{\sf St}^{D_i})$ that accepts exactly the  set of sequences of observations $\eta \in ( O_i \times (\Sigma \cup \{ \tau \}))^{\omega}$ that are {\em doomsday compatible} for Player~$i$. We know that the number of states in $D_i$ is ${\bf O}(2^{|S|^2 \log |S|})$ and the number of Streett pairs is bounded by
          ${\bf O}(|S|^2\cdot n)$, where $|S|$ is the number of states in $G$. 
	\item Second, we consider a turn-based game played on $D_i$ by two players, {\sf A} and {\sf B}, that move a token from states to states along edges of $D_i$ as follows:
		\begin{enumerate}
			\item initially, the token is in some state $q$
			\item then in each round: {\sf B} chooses an observation $o\in O_i$ in the set $\{ o\in O_i\ |\ \exists (q,(o,\sigma),q')\in \delta^{D_i}\}$. Then {\sf A} chooses 
                          a transition $(q,(o,\sigma),q')\in \delta^{D_i}$ (which is completely determined by $\sigma$ as $D_i$ is deterministic), and  the token is moved to
                          $q'$ where a new round starts. 
                      \end{enumerate}
		The objective of {\sf A} is to enforce from state $q$ an infinite sequence of states, so a run of $D_i$ that starts in $q$, and which satisfies ${\sf St}^{D_i}$ the Streett condition of $D_i$.  For each $q$, this can be decided in time polynomial in the number of states in $D_i$ and exponential in the number of Streett pairs in ${\sf St}^{D_i}$, see~\cite{DBLP:conf/lics/PitermanP06} for an algorithm with the best known complexity. Thus, the overall complexity is exponential in the size of the game structure $G$. We denote by ${\sf Win} \subseteq Q^{D_i}$ the set of states $q$ from which {\sf A} can win the game above.
		\item Note that if $(o_1,\sigma_1)\dots (o_m,\sigma_m)$ is the trace of a path from $q_{{\sf init}}$ in $D_i$ to a state $q \in {\sf Win}$, then clearly 
                  $(o_1,\sigma_1)\dots (o_{n-1},\sigma_{n-1})o_n$ is good for retaliation. Indeed, the winning strategy of {\sf A} in $q$ is an observation based retaliating strategy $\lambda^R_i$ for Player~$i$ in $G$. 
                  On the other hand, if a prefix of observations reaches $q \not\in {\sf Win}$ then by determinacy of Streett games, we know that {\sf B} has a winning strategy in $q$
                  and this winning strategy is a strategy for the coalition (against Player~$i$) in $G$ to enforce a play in which Player~$i$ does not win and at least one of the other players wins.
                  So, from $D_i$ and ${\sf Win}$, we can construct a finite state automaton $C_i$ which is obtained as a copy of $D_i$ with the following acceptance condition: 
                  a prefix $\kappa=(o_0,\sigma_0),(o_1,\sigma_1),\dots,(o_{k-1},\sigma_{k-1}),o_k$ is accepted by $C_i$ if there exists a path $q_0 q_1 \dots q_{k}$ in $C_i$ such
that $q_0$ is the initial state of $C_i$ and either there exists a transition labeled $(o_k,\sigma)$ from $q_k$ to a state of ${\sf Win}$.\qed

% and all the transitions that leaves $q_{k-1}$ are labelled by a pair in $(o_k,\sigma)$ for some $\sigma \in \Sigma \cup \{ \tau \}$.\qed
  \end{itemize}
\end{proof}

\makeatletter
\let\@savethebibliography\thebibliography
\def\thebibliography#1{%
  \@savethebibliography{#1}
  \global\def\c@enumiv{\value{suitebib}}}
\makeatother
\putbib[biblio]
\end{bibunit}

% \bibliographystyle{plain}
% \bibliography{bib2}

\end{document}